\documentclass[journal]{IEEEtran}
\ifCLASSINFOpdf
\else
\fi
\usepackage{graphicx}
\usepackage{amsfonts}
\usepackage{cite}
\usepackage{color}
\usepackage{float}  
\usepackage{subfigure} 

\usepackage[cmex10]{amsmath}
\usepackage{amssymb}
\usepackage{amsthm}
\usepackage{mathrsfs}
\usepackage{bm}
\usepackage[mathscr]{eucal}

\usepackage[short]{optidef}

\newtheorem{theorem}{Theorem}
\newtheorem{corollary}{Corollary}

\usepackage{algorithm}
\usepackage{algorithmic}

\usepackage[hidelinks]{hyperref}

\begin{document}
%
\title{RIS Empowered Near-Field Covert Communications}
%
%
%

\author{Jun Liu,
       Gang Yang,~\IEEEmembership{Member,~IEEE,}
     Yuanwei Liu, \IEEEmembership{Fellow, IEEE,} 
     and Xiangyun Zhou,  \IEEEmembership{Fellow, IEEE}
\thanks{This work was supported by the Shenzhen Science and Technology Program under Grants JCYJ20220530164814032 and JCYJ20220818103201004.}
\thanks{Jun Liu is with the National Key Laboratory of Wireless Communications, University of Electronic Science and Technology of China (UESTC), Chengdu 611731, China (e-mail: junl@std.uestc.edu.cn). }
\thanks{Gang Yang is with the Shenzhen Institute for Advanced Study, and the National Key Laboratory of Wireless Communications, University of Electronic Science and Technology of China (UESTC), Chengdu 611731, China (e-mail: yanggang@uestc.edu.cn). (\emph{Corresponding author: Gang Yang})}
\thanks{Yuanwei Liu is with the Department of Electrical and Electronic Engineering, The University of Hong Kong, Hong Kong (e-mail: yuanwei@hku.hk).}
\thanks{Xiangyun Zhou is with the School of Engineering, The Australian National University, Canberra, ACT 2600, Australia (e-mail: xiangyun.zhou@anu.edu.au).}}
\maketitle

\begin{abstract}
This paper studies an extremely large-scale reconfigurable intelligent surface (XL-RIS) empowered covert communication system in the near-field region. Alice covertly transmits messages to Bob with the assistance of the XL-RIS, while evading detection by Willie. To enhance the covert communication performance, we maximize the achievable covert rate by jointly optimizing the hybrid analog and digital beamformers at Alice, as well as the reflection coefficient matrix at the XL-RIS. An alternating optimization algorithm is proposed to solve the joint beamforming design problem. For the hybrid beamformer design, a semi-closed-form solution for fully digital beamformer is first obtained by a weighted minimum mean-square error based algorithm, then the baseband digital and analog beamformers at Alice are designed by approximating the fully digital beamformer via manifold optimization. For the XL-RIS's reflection coefficient matrix design, a low-complexity alternating direction method of multipliers based algorithm is proposed to address the challenge of large-scale variables and unit-modulus constraints. Numerical results unveil that i) the near-field communications can achieve a higher covert rate than the far-field covert communications in general, and still realize covert transmission even if Willie is located at the same direction as Bob and closer to the XL-RIS; ii) the proposed algorithm can enhance the covert rate significantly compared to the benchmark schemes; iii) the proposed algorithm leads to a beam diffraction pattern that can bypass Willie and achieve high-rate covert transmission to Bob.


\end{abstract}


%
\IEEEpeerreviewmaketitle

\section{Introduction}
\IEEEPARstart{C}{overt} communication is an emerging and cutting-edge technology, which aims at hiding the communication devices and/or communication behavior from being detected by a watchful warden. As a pioneering work, the authors in~\cite{bash2013limits} established the ``square-root-law'' for the fundamental limits of covert communication, which claims that at most $\mathcal{O}(\sqrt{n})$ bits of information can be transmitted covertly and reliably over $n$ channel uses in the additive white Gaussian noise (AWGN) channel. According to this law, the covert rate for per channel use is asymptotic to zero. Fortunately, several schemes have been proposed to achieve positive covert communication rates, e.g., multiple antennas~\cite{zheng2019multi}, noise uncertainty~\cite{he2017covert}, channel uncertainty~\cite{wang2018covert}, full-duplex receiver~\cite{shahzad2018achieving} and relay~\cite{hu2018covert, hu2019transmission}. However, these methods~\cite{zheng2019multi,soltani2018covert,he2017covert,wang2018covert,shahzad2018achieving,hu2018covert, hu2019transmission} are typically channel-adaptive, thereby making covert communication highly dependent on the wireless propagation environment.

To tackle this issue, several studies have integrated reconfigurable intelligent surfaces (RIS) into covert communication systems~\cite{wang2021intelligent,si2021covert,zhou2021intelligent,wang2022covert,lv2021covert}. RIS is composed of a large number of passive reflection elements and can manipulate the wireless propagation environment in a programmatic manner. The main objective of introducing RIS into wireless communication systems is to address the base stations' blockage issues and enhance the network coverage~\cite{Liaskos2018wireless, yang2021reconfigurable}. The authors of~\cite{wang2021intelligent} demonstrated the great potential of integrating RIS into covert communications. The active and passive beamformers were jointly designed in~\cite{si2021covert} for RIS-assisted covert communication systems. In order to meet the low-latency application requirement, the authors of~\cite{zhou2021intelligent} studied RIS-assisted covert communication systems with a finite number of channel uses. The RIS-assisted covert communication systems were also studied with other enabling technologies such as unmanned aerial vehicle~\cite{hu2018covert} and non-orthogonal multiple access~\cite{lv2021covert}. 

However, there are some weaknesses in the existing works on RIS-assisted covert communications. First, the RIS provides limited covert-rate performance enhancement due to the double fading effect in RIS-assisted covert systems, as it is equipped with a small number of reflection elements~\cite{di2020smart}. This problem becomes more significant at higher frequency bands with severe path loss, such as millimeter wave (mmWave) and terahertz (THz). Second, these works focus on the \textit{planar-wave} propagation in far-field communications, resulting in limited covert rate performance in scenarios with spatial correlation. When Willie and Bob are located at the same angle with respect to Alice, it is almost impossible to distinguish them under the planar-wave model. Consequently, Alice's transmission power is limited to keep the received power at Willie below a certain level, hindering the improvement of covert communication performance. Third, most existing works focused on fully digital architecture at Alice, which are not suitable for covert communication systems with massive antennas. Fourth, these works~\cite{wang2021intelligent,si2021covert,zhou2021intelligent,wang2022covert,lv2021covert} assume that Bob and Willie are equipped with a single antenna. In practical scenarios, Willie is more likely to adopt multi-antenna technology to enhance its detection performance.

It is noticed that the extremely large-scale RIS (XL-RIS) was recently proposed to compensate for the severe double-fading effect in the cascaded channel~\cite{mu2023ris,liu2023low,yu2023channel,shen2023multi,miridakis2022zero,han2022localization}. On the one hand, the XL-RIS achieves significant signal enhancement by enlarging the array aperture with a massive number of low-cost reflection elements. On the other hand, the XL-RIS provides great potential for enhancing the covert performance in the near-field region. As both the array aperture and operating frequencies increase, wireless communications inevitably operate in the near-field region. The Rayleigh distance that is commonly used to distinguish between near and far fields~\cite{liu2023near}, is given by $\frac{2D^2}{\lambda}$, with $D$ and $\lambda$ denoting the array aperture and the wavelength, respectively. In contrast to the conventional far-field plane-wave channel model, the electromagnetic propagation in near field is characterized by the \textit{spherical-wave} channel model that depends on both angle and distance~\cite{liu2023near, wang2023near}, providing the additional distance degree of freedom (DoF) for wireless communications design. The near-field transmission enables the beam pattern to focus on a specific area, i.e., \textit{beam focusing}~\cite{zhang2022beam,zhang20236g,zhang2023physical}, thus the XL-RIS is promising for achieving efficient covert communications. 

Despite that some works have studied the beamforming design in XL-RIS-assisted systems~\cite{liu2023low, shen2023multi}, the beamforming design in covert communication system has different characteristic from these works. In the scenario considered in~\cite{liu2023low, shen2023multi}, all users are legitimate and only require the joint design of active beamforming at the BS and passive beamforming at the XL-RIS to enhance the performance of legitimate users. However, in a covert communication system, beamforming design not only improves the performance of legitimate users but also suppresses the power leakage to the warden, Willie. Consequently, the beamforming design should balance the ability to enhance legitimate-user-link gain and suppress the power leakage, thus the problem is more intricate in covert communication systems.


 %
 
 

To the best of our knowledge, the XL-RIS empowered near-field covert communications have not been studied yet in the literature. In this paper, we investigate the performance analysis and optimization for an XL-RIS empowered covert communication system operating in the near-field region. The main contributions are as follows: 

\begin{itemize}
    \item We explore an XL-RIS empowered near-field covert communication system, consisting of a legitimate transmitter named Alice, a covert receiver named Bob, a warden receiver named Willie, and an XL-RIS. Alice adopts hybrid beamforming for high-frequency transmission, and its direct link to either Bob or Willie is blocked. Both Bob and Willie are within the near-field region of the XL-RIS. The detection error probability (DEP) performance is first analyzed for Willie. The optimal detection threshold is further derived in closed form, and the resulting minimal DEP is exploited to obtain the desired covertness constraint for system optimization. 

    \item An optimization problem is formulated to maximize the achievable covert rate at Bob. Specifically, this problem jointly optimizes Alice's baseband digital beamformer and analog beamformer, as well as the XL-RIS's reflection coefficient matrix, subject to Alice's transmit power constraint and analog beamformer unit-modulus constraints, Willie's covertness constraint, and XL-RIS's reflection coefficient constraints. The optimization problem features large-scale coupled variables and unit-modulus constraints, thus is non-convex and challenging to solve in general. 
        
    \item An alternating optimization (AO)-based algorithm is proposed to decompose the original problem into two sub-problems. For Alice's hybrid beamformer design subproblem, a two-stage algorithm is proposed. In the first stage, a semi-closed-form solution for fully digital beamformer is derived by a weighted minimum mean-square error (WMMSE)-based algorithm. In the second stage, the baseband digital and analog beamformers are obtained by approximating the fully digital beamformer via manifold optimization (MO). For the XL-RIS's reflection coefficient design subproblem, a low-complexity alternating direction method of multipliers (ADMM)-based algorithm is proposed to address the challenge of large-scale variables and unit-modulus constraints. Also, both the convergence and computational complexity of the proposed algorithm are analyzed. 

    \item Numerical results unveil that 1) the near-field communications in general achieve a higher covert rate than the far-field covert communications, and still realize covert transmission even if Willie is located at the same direction as Bob and closer to the XL-RIS; 2) the proposed algorithm can improve the achievable covert rate compared to the benchmark schemes such as random-phase and zero-forcing; 3) the proposed algorithm leads to a new beam diffraction pattern, i.e., a transitional state between beam steering and beam focusing, which can split and bypass the non-target user Willie, then converge on the intended user Bob for high-rate covert transmission.
\end{itemize}



The rest of this work is organized as follows. Section~\ref{Sec:Sys_model} presents the system model for the XL-RIS empowered near-field covert communication system, including the system description, channel model, and signal model. Section~\ref{Sub:Det_per_anay} analyzes the DEP performance at Willie. Section~\ref{Sub:cov_comm_deign} formulates an achievable covert rate maximization problem, proposes an AO-based solving algorithm, and provides its convergence and computational complexity analyses. Section~\ref{Sub:Nem_res} gives numerical results. Section~\ref{Sub:Conc} concludes this work.

\textit{Notations:} The main notations are listed as follows. $\mathcal{CN}\left(\mu,\sigma^2 \right)$ denotes the the circularly symmetric complex Gaussian distribution with mean $\mu$ and variance $\sigma^2$. $\mathbf{I}_M$ denotes the $M\times M$ identity matrix. For any vector $\mathbf{v}$, $v_i$ denote the $i$-th element, and $||\mathbf{v}||$ denotes the Euclidean norm. $\text{diag}(\mathbf{v})$ denotes the diagonal operation. For any matrix $\mathbf{V}$, the transpose, conjugate, and conjugate transpose are $\mathbf{V}^{\sf T}$, $\mathbf{V}^*$ and $\mathbf{V}^{\sf H}$, respectively, and $v_{i,j}$ is the $i$-th row and $j$-th column element. $||\mathbf{V}||_{\sf F}$ denotes the Frobenius norm. $|\mathbf{V}|$ denotes the determinant of $\mathbf{V}$. $\text{vec}(\mathbf{V})$ and $\text{Tr}(\mathbf{V})$ denote the vectorization operation and trace of the matrix. $|x|$ denotes the absolute value of $x$, and $\Re\{x\}$ is its real part.  $\circ$ and $\otimes$ denote the Hadamard product and Kronecker product, respectively.

\section{System Model} \label{Sec:Sys_model}
In this section, we present the system description, far-field and near-field channel models, and signal model for the XL-RIS empowered near-field covert communication system.

\begin{figure}[t]
    \centering
    \includegraphics[width=0.80\linewidth]{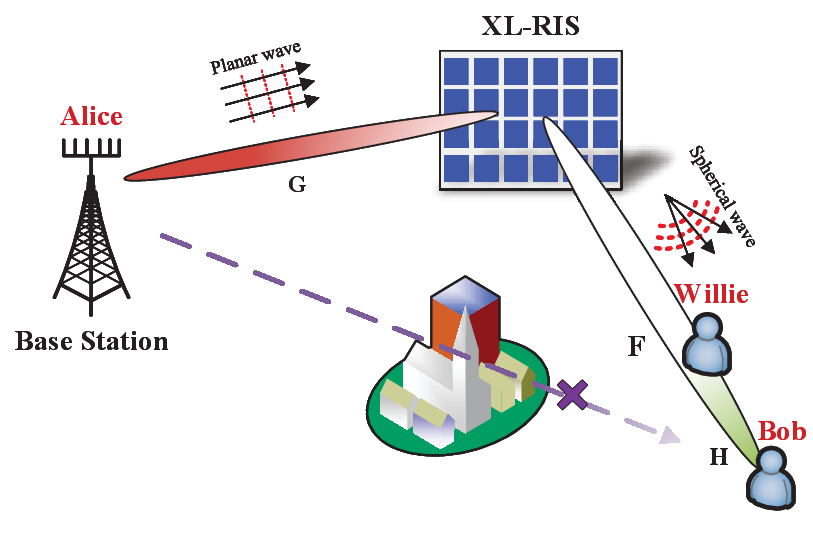}
    \caption{Illustration of an XL-RIS empowered near-field covert communication system.} 
    \label{fig:system}
\end{figure}

\subsection{System Description}\label{subs:sys_des}
As depicted in Fig.~\ref{fig:system}, we consider an XL-RIS empowered near-field covert communication system, which consists of a legitimate transmitter (referred to as Alice) equipped with $M_{\sf A} \ (M_{\sf A} \geqslant 1)$ antennas, a covert receiver (referred to as Bob) equipped with $M_{\sf B} \ (M_{\sf B}\geqslant 1)$ antennas, and a warden receiver (referred to as Willie) equipped with $M_{\sf W} \ (M_{\sf W} \geqslant 1)$ antennas. Alice operates in the high-frequency band, such as mmWave or THz, with a large number of antennas. Therefore, the hybrid beamforming architecture is adopted at Alice~\cite{sohrabi2016hybrid}. In the meanwhile, fully-digital antenna architecture is adopted at Bob for simplicity. Due to blockage in areas with lots of obstacles~\cite{wang2021intelligent,chen2021enhancing}, the direct links between Alice and both Bob and Willie are unavailable~\footnote{The detection performance analysis at Willie in Section~\ref{Sub:Det_per_anay} and the joint beamforming design for covert communication in Section~\ref{Sub:cov_comm_deign} can be extended to the scenario when Alice and Bob as well as Willie have weak direct links, as studied in~\cite{wang2021intelligent,si2021covert,zhou2021intelligent}.}. The XL-RIS is utilized to enable the covert communication from Alice to Bob. From the criterion of RIS deployment~\cite{Liaskos2018wireless}, the XL-RIS is deployed to be closer to Bob, thus is assumed to be located in the far-field region of Alice~\cite{cheng2023achievable}. A uniform planar array (UPA) is employed at the XL-RIS, which is has $N_{\sf y}$  horizontal rows and $N_{\sf z}$ vertical columns ($N = N_{\sf y} \times N_{\sf z}$).

The distance between the XL-RIS and both Bob and Willie is shorter than the Rayleigh distance $d_{\sf R} = \frac{2 f (D+D_{\sf B})^2}{c}$, where $f$ represents the carrier frequency, $c$ is the light speed, and $D$ and $D_{\sf B}$ are antenna apertures of the XL-RIS and Bob, respectively. In other words, Bob and Willie are assumed to be located within the near-field region of the XL-RIS~\footnote{It is most valuable to consider that Bob and Willie are in the near-field region. Otherwise, if Willie is in the far-field region while Bob is in the near-field region, Willie's received power is likely to be lower than the noise power due to the double-fading effect in the mmWave band, which is a meaningless scenario. On the contrary, if Bob is in the far-field region while Willie is in the near-field region, the scenario degenerates into the RIS-assisted far-field covert communication, which has been studied by a large number of literatures~\cite{wang2021intelligent,si2021covert,zhou2021intelligent,wang2022covert,lv2021covert}.}. This assumption is justified for the  XL-RIS empowered system from the following two aspects. First, the XL-RIS inherently has a larger antenna aperture, and $d_{\sf R}$ is proportional to the square of $D$. Second, future wireless networks are expected to operate at tremendously high-frequency bands, where $d_{\sf R}$ is proportional to the carrier frequency $f$. Consequently, the reflection signals from the XL-RIS to both Bob and Willie exhibit spherical propagation.

For Willie, the receiver noise power fluctuates due to the variations of wireless communication environment factors, such as temperature, humidity and weather, and the receiver also suffers from calibration errors when measuring the noise power. Therefore, this paper takes into account the noise uncertainty at Willie receiver, and adopts the bounded uncertainty model~\cite{goeckel2015covert,ta2019covert}. Specifically, the exact noise power at Willie $\sigma_{\sf W}^2$ in dB scale is uniformly distributed within $\left[ {\widehat \sigma _{\sf{W}\text{,dB}}^2 - {\rho _{\text{dB}}},\widehat \sigma _{\sf{W}\text{,dB}}^2 + {\rho _{\text{dB}}}} \right]$, where $\widehat \sigma _{\sf{W}\text{, dB}}^2=10\log_{10}(\widehat \sigma_{\sf{W}}^2)$  and ${\rho _{\text{dB}}} = 10\log_{10}(\rho)$ represent the nominal noise power and the degree of noise uncertainty in dB scale. The probability density function (PDF) of $\sigma_{\sf W}^2$ for the bounded uncertainty model is given by
\begin{equation}\label{eq:np_pdf}
    {f_{\sigma _{\sf{W}}^2}} = \left\{ {\begin{array}{*{20}{l}}
  {\frac{1}{{2\ln \left( \rho  \right)x}},}&{{\text{if }}\frac{1}{\rho }\widehat \sigma _{\text{W}}^2 \leqslant x \leqslant \rho \widehat \sigma _{\sf{W}}^2}, \\ 
  {0,}&{{\text{otherwise}}}. 
\end{array}} \right.
\end{equation}



\subsection{Channel Model}
Let ${\mathbf{G}} \in {\mathbb{C}^{N \times M_{\sf A}}}$ denote the channel matrix between Alice and the XL-RIS, ${\mathbf{H}} \in {\mathbb{C}^{{M_{\sf{W}}} \times N}}$ denote the channel matrix between XL-RIS and Bob, and ${\mathbf{F}} \in {\mathbb{C}^{{M_{\sf{W}}} \times N}}$ denote the channel matrix between XL-RIS and Willie. The phase-shift matrix of the XL-RIS is defined as a diagonal matrix ${\mathbf{\Theta }} = {\sf{diag}}\left( {{e^{\jmath{\theta _1}}}, \ldots , {e^{\jmath{\theta _N}}}} \right)$, where $\theta_n \in [0, 2\pi)$ is the phase shift of the $n$-th reflection element.

The planar-wave-based far-field channel model is employed to characterize the communication channel between Alice and XL-RIS. Based on the Saleh-Valenzuela (SV) channel model, the Alice-to-XL-RIS channel $\mathbf{G}$ can be given as
\begin{equation}
    \mathbf{G} = \sqrt {\frac{{M_{\sf A}N{\chi _{{\sf{ar}}}}}}{{{L_{\sf{p}}}}}} \sum\limits_{i = 1}^{{L_{\sf{p}}}} {{\alpha _i}{{\mathbf{a}}_{{\sf{UPA}}}}\left( {\vartheta _{{\sf{ar,r}}}^i,\phi _{{\sf{ar,r}}}^i} \right)\mathbf{a}_{{\sf{ULA}}}^{\sf H}\left( {\vartheta _{{\sf{ar,t}}}^i} \right)},
\end{equation}
where $L_{\sf p}$ is the total number of resolvable signal paths, ${\chi _{{\sf{ar}}}}$ represents the average path loss, $\alpha_i$ is the complex channel gain of the $i$-th path, $\vartheta_{\sf ar,r}^i$, and $\phi_{\sf ar,r}^i$ denote the azimuth and elevation angles of arrival (AoAs) associated with the XL-RIS, ${\vartheta _{{\sf{ar,t}}}^i}$ represents the angle of departure (AoD) associated with Alice, and $\mathbf{a}_{\sf UPA}$ and $\mathbf{a}_{\sf ULA}$ are the normalized array response vectors associated with the UPA and the uniform linear array (ULA), respectively. Specifically, the normalized array response vector for ULA with $M$ elements is given by
\begin{equation}
\begin{aligned}
         &\mathbf{a}_{\sf ULA} \left( \vartheta\right)  =\\ &  \frac{1}{\sqrt{M}} \left[1,\ldots,e^{\jmath\frac{2\pi d}{\lambda}(m-1)\sin(\vartheta)},\ldots, e^{\jmath\frac{2\pi d}{\lambda}(M-1)\sin(\vartheta)} \right]^{\sf T},
\end{aligned}
\end{equation}
where $d$ is the antenna spacing. For UPA, the normalized array response vector with $N=N_{\sf y}\times N_{\sf z}$ elements is given by
\begin{equation}
\begin{aligned}
    &   \mathbf{a}_{\sf UPA} \left( \vartheta,\phi\right)  = \\  & \frac{1}{\sqrt{N}}  [1,e^{\jmath\frac{2\pi d}{\lambda}\sin(\vartheta)\cos(\phi)}, \ldots, e^{\jmath\frac{2\pi d}{\lambda}(N_{\sf y}-1)\sin(\vartheta)\cos(\phi) } ]^{\sf T} 
    \\ & ~~~~\otimes [1,e^{\jmath\frac{2\pi d}{\lambda}\sin(\phi)}, \ldots, e^{\jmath\frac{2\pi d}{\lambda}(N_{\sf z}-1)\sin(\phi) } ]^{\sf T}.
\end{aligned}
\end{equation} 


The wavefront of electromagnetic waves in the near-field region is modeled as spherical waves. A three-dimensional Cartesian coordinate system is considered, where the XL-RIS is deployed on the $yz$-plane with its center located at the origin of coordinate $(0,0,0)$. Therefore, the element in $n_1$-th column and $n_2$-th row of the XL-RIS is located as $(0,n_{1}d,n_{2}d)$, where $n_1=\frac{-N_{\sf y}+1}{2},\ldots,\frac{N_{\sf y}-1}{2}$ and $n_2=\frac{-N_{\sf z}+1}{2},\ldots,\frac{N_{\sf z}-1}{2}$. Similarly, the coordinates of the $m_{\sf B}$-th element at Bob and $m_{\sf W}$-th element at Willie can be given as $(x_{\sf B},y_{\sf B}+{\tilde{m}_{\sf B}}d,0)$ and $(x_{\sf W},y_{\sf W}+{\tilde{m}_{\sf W}}d,0)$, respectively, where ${\tilde{m}_{\sf B}}=m_{\sf B}-\frac{M_{\sf B}-1}{2}$ and ${\tilde{m}_{\sf W}}=m_{\sf W}-\frac{M_{\sf W}-1}{2}$. Since the XL-RIS is deployed near Bob, the line-of-sight (LoS) near-field channel between the XL-RIS and Bob is assumed, which can be modeled as 
\begin{equation}
    {\mathbf{H}} = {\left[ {{{\mathbf{h}}_1}, \ldots ,{{\mathbf{h}}_{{M_{\sf{B}}}}}} \right]^{\sf{T}}},
\end{equation}
where $ {{\mathbf{h}}_{{m_{\sf{B}}}}} = (1\big/\sqrt{N})\ \big[ {{\chi _{{m_{\sf{B}}},1}}{e^{ - \jmath \frac{{2\pi f}}{c}\left( {{r_{{m_{\sf{B}}},1}} ~- ~{r_{{m_{\sf{B}}}}}} \right)}}, ~ \ldots, } \\ {{\chi _{{m_{\sf{B}}},N}}{e^{ - \jmath \frac{{2\pi f}}{c}\left( {{r_{{m_{\sf{B}}},N}} - {r_{{m_{\sf{B}}}}}} \right)}}} \big]^{\sf{T}}$. More specifically, $r_{m_{\sf B},n}$ represents the distance between the $n$-th ($n = (n_1-1)N_{\sf y}+n_2$) element of the XL-RIS and the $m_{\sf B}$-th element of Bob, $r_{m_{\sf B}} $ denotes the reference distance from $(0,0,0)$ to  $(x_{\sf B},y_{\sf B}+{\tilde{m}_{\sf B}}d,0)$, and ${\chi _{{m_{\sf{B}}},n}} = \frac{c}{4\pi fr_{m_{\sf B},n}}$ is the free-space large-scale path loss between $n$-th element of XL-RIS and the $m_{\sf B}$-th element of Bob. The near-field channel between the XL-RIS and Willie $\mathbf{F}$ can be obtained by following the similar channel model.


Moreover, we assume that Willie’s perfect channel state information (CSI) is available at Alice. We make this assumption based on the following reasons. First, from the practical perspective, our proposed algorithm can be applied in two typical situations: i) when Willie is an unlicensed user, Willie usually uses a radiometer to detect the presence of transmission between Alice and Bob. It is known that almost all practical radiometers are superheterodyne receivers whose signal leakage is unavoidable~\cite{chaman2018ghostbuster}. Hence, as proposed in~\cite{wang2021intelligent}, Alice can obtain Willie-involved CSI through some advanced detecting equipment like “Ghostbuster” introduced in~\cite{chaman2018ghostbuster}. ii) when Willie is a licensed user, however, Alice suspects Willie as a malicious user during the current transmission to Bob~\cite{wang2021intelligent,zhou2019joint}. In this situation, Willie’s perfect CSI is available at Alice. Second, from the theoretical perspective, almost all literature on RIS-assisted covert communications like~\cite{chen2021enhancing,zhou2021intelligent,si2021covert} assumes that Willie-involved channels can be obtained, which can provide an upper bound on the achievable covert rate performance. It is worth noting that the method proposed in this paper can be extended to the scenarios of imperfect or partial Willie CSI available at Alice, by utilizing existing methods such as $\mathcal{S}$-procedure~\cite{wang2021Energy}, Bernstein-type inequality~\cite{sun2022outage} and triangle inequality~\cite{zhou2021intelligent}. The scenarios of imperfect or statistical Willie CSI at Alice are beyond the scope of this work, due to limited space.

\subsection{Signal Model}

Assuming that Alice transmits $L$ data streams to Bob with $M_{\sf A}^{\sf{RF}}$ ($L \leqslant M_{\sf A}^{\sf{RF}} \leqslant M_{\sf A}$) transmit RF chains, the transmitted signal at Alice with hybrid beamforming can be expressed as 
\begin{equation}
    {\mathbf{x}} = {{\mathbf{W}}_{{\text{RF}}}}{{\mathbf{W}}_{{\text{BB}}}}{\mathbf{s}},
\end{equation}
where ${{\mathbf{W}}_{{\sf{RF}}}} \in {\mathbb{C}^{M \times M_{\sf A}^{\sf{RF}}}}$ denotes the analog precoding matrix, ${{\mathbf{W}}_{{\sf{BB}}}} \in {\mathbb{C}^{{M_{\sf{A}}^{\sf{RF}}} \times L}}$ denotes the baseband digital precoding matrix, and $\mathbf{s}\in \mathbb{C}^{L\times 1}$ with ${\mathbb{E}\left( {{\mathbf{s}}{{\mathbf{s}}^{\sf H}}} \right) = {{\mathbf{I}}_L}}$ represents the data transmitted to Bob. The $i$-th row and the $j$-th column element of ${{\mathbf{W}}_{{\sf{RF}}}}$, denoted by $w_{{\sf{RF}}}^{\left( {i,j} \right)}$, satisfies the following unit-modulus constraint
\begin{equation}
    w_{{\sf{RF}}}^{\left( {i,j} \right)} \in \mathcal{F} \triangleq \left\{ {{e^{\jmath\varphi }}|\varphi  \in \left( {0,2\pi } \right]} \right\}.
\end{equation}
The received signal at Bob with fully digital combiner is
\begin{equation}
    {{\mathbf{y}}_{\sf{B}}} = {\mathbf{H\Theta G}}{{\mathbf{W}}_{{\sf{RF}}}}{{\mathbf{W}}_{{\sf{BB}}}}{\mathbf{s}} + {{\mathbf{n}}_{\sf{B}}},
\end{equation}
where ${{\mathbf{n}}_{\sf{B}}}(i) \sim \mathcal{C}\mathcal{N}\left( {{\mathbf{0}},\sigma _{\sf{B}}^2{{\mathbf{I}}_{{M_{\sf{B}}}}}} \right)$ is the AWGN at Bob. Accordingly, the mutual information between Alice and Bob is 
\begin{equation} 
\begin{aligned}
    &C_{\sf{B}}\left({{\bf{W}}_{{\sf{RF}}}},{{\bf{W}}_{{\sf{BB}}}},{\bf{\Theta }} \right) = \\ \label{eq:capacity}
    &{\log _2}\left| {{{\mathbf{I}}_{{M_{\sf{B}}}}} + \sigma _{\sf{B}}^{ - 2}{\mathbf{H\Theta G}}{{\mathbf{W}}_{{\sf{RF}}}}{{\mathbf{W}}_{{\sf{BB}}}}{{\mathbf{W}}_{{\sf{BB}}}^{\sf H}}{{\mathbf{W}}_{{\sf{RF}}}^{\sf H}}{{\mathbf{G}}^{\sf H}}{{\mathbf{\Theta }}^{\sf H}}{{\mathbf{H}}^{\sf H}}} \right|.
\end{aligned}    
\end{equation}
It is notice that the Bob's noise variance is assumed to be a constant $\sigma _{\sf{B}}^{2}$ in~\eqref{eq:capacity}. This simplified assumption does not affect the beamforming optimization design in latter Section~\ref{Sub:cov_comm_deign}, as Bob's noise power only appears in the objective function of the formulated covert-rate maximization problem~\eqref{eq:overall_pro}.

\section{Detection Performance Analysis at Willie}\label{Sub:Det_per_anay}
In the considered covert communication system, Willie attempts to detect the presence of transmission from Alice by conducting a binary hypothesis test. This section presents Willie's detection method and analyzes its DEP performance.
\subsection{Detection of Covert Communications at Willie}\label{subsec:DCCW}
Willie needs to discriminate between two hypotheses: the null hypothesis $\mathcal{H}_0$ that Alice remains silent, and the alternative hypothesis $\mathcal{H}_1$ that Alice is transmitting signals. We assume that Willie observes $K$ samples within a time slot. Considering these two hypotheses, the received signal at Willie for the $i$-th sample can be expressed as follows
\begin{equation}
    {{\mathbf{y}}_{\sf{W}}}(i) = \left\{ {\begin{array}{*{20}{l}}
  {{{\mathbf{n}}_{\sf{W}}}(i),}&{\text{if}~{\mathcal{H}_0}} \\ 
  {{\mathbf{F\Theta G}}{{\mathbf{W}}_{{\sf{RF}}}}{{\mathbf{W}}_{{\sf{BB}}}}{\mathbf{s}}(i) + {{\mathbf{n}}_{\sf{W}}}(i),}&{\text{if}~{\mathcal{H}_1}} 
\end{array}} \right.
\end{equation}
where  ${{\mathbf{n}}_{\sf{W}}}(i) \sim \mathcal{C}\mathcal{N}\left( {{\mathbf{0}},\sigma _{\sf{W}}^2{{\mathbf{I}}_{{M_{\sf{W}}}}}} \right)$ is the AWGN at Willie.

In covert communications, Willie's objective is to determine whether the observation samples $\left\{ \mathbf{y}_{\sf W} (i) \right\}_{i = 1}^K$ are purely background noises, denoted as $\mathcal{H}_0$, or the signals comprising both Alice's signals and the background noises, denoted as $\mathcal{H}_1$, with the minimum DEP. To this end, according to the Neyman-Pearson criterion, Willie's optimal test is the likelihood ratio test~\cite[Ch. 3.3]{kay1993fundamentals}, i.e., $\Upsilon \left( {\left\{ {{{\mathbf{y}}_{\sf W}}(i)} \right\}} \right) = \frac{{f\left( {\left\{ {{{\mathbf{y}}_{\sf W}}(i)} \right\}|{\mathcal{H}_1}} \right)}}{{f\left( {\left\{ {{{\mathbf{y}}_{\sf W}}(i)} \right\}|{\mathcal{H}_0}} \right)}}\mathop {\mathop  \gtrless \limits^{{\mathcal{D}_1}} }\limits_{{\mathcal{D}_0}} \gamma$, where $\mathcal{D}_0$ and $\mathcal{D}_1$ denote Willie's decision in favor of hypotheses  $\mathcal{H}_0$ and $\mathcal{H}_1$, and ${f\left( {\left\{ {{{\mathbf{y}}_{\sf W}}(i)} \right\}|{\mathcal{H}_0}} \right)}$ and ${f\left( {\left\{ {{{\mathbf{y}}_{\sf W}}(i)} \right\}|{\mathcal{H}_1}} \right)}$ are the likelihood functions of Willies's signal samples under hypothesis $\mathcal{H}_0$ and $\mathcal{H}_1$, respectively. However, in typical application scenarios, Willie does not have the instantaneous CSI of $\mathbf{G}$ and $\mathbf{F}$, making it difficult to derive the optimal detector. Therefore, we assume that Willie employs a radiometer detector~\cite{chen2021enhancing,zheng2019multi,he2018covert}, which does not require instantaneous CSI and thus is commonly used in practice. As a result, the decision rule for a radiometer detector at Willie is 
\begin{equation}\label{eq:Det_rule}
    {T_{\sf{W}}}  \mathop {\mathop  \gtrless \limits^{{\mathcal{D}_1}} }\limits_{{\mathcal{D}_0}} \Gamma,
\end{equation}
where ${T_{\sf{W}}} \triangleq \frac{1}{K}\sum\limits_{i = 1}^K {{{\left\| {{{\mathbf{y}}_{\text{W}}}(i)} \right\|}^2}}$ is the test statistic, which represents the averaged power received by Willie with equal gain combiner, and $\Gamma >0$ denotes the Willie's detection threshold. Following the widely adopted assumption in existing work on covert communications~\cite{bash2013limits,wang2021intelligent,si2021covert}, we consider that Willie exploits an infinite number of signal samples in a time slot to perform detection, i.e., $K \to \infty$. Consequently, the averaged received power  ${T_{\sf{W}}}$ at Willie is obtained as 
\begin{equation} \label{eq:RP_hyp}
    {T_{\sf{W}}} = \left\{ {\begin{array}{*{20}{l}}
  {{M_{\sf{W}}}\sigma _{\sf{W}}^2,}&{\text{if}~{\mathcal{H}_0}} \\ 
  {\left\| {{\mathbf{F\Theta G}}{{\mathbf{W}}_{{\sf{RF}}}}{{\mathbf{W}}_{{\sf{BB}}}}} \right\|_{\sf{F}}^2 + {M_{\sf{W}}}\sigma _{\sf{W}}^2,}&{\text{if}~{\mathcal{H}_1}} 
\end{array}} \right.
\end{equation}

Assuming equal prior probabilities for hypotheses $\mathcal{H}_0$ and $\mathcal{H}_1$, the performance of Willies's hypothesis test is measured by the DEP, denoted by $\xi$ ($0\leqslant \xi \leqslant 1$), which is defined as 
\begin{equation}
    \xi = P_{\sf{FA}} +P_{\sf{MD}}, 
\end{equation}
where $P_{\sf{FA}} = \Pr(\mathcal{D}_1|\mathcal{H}_0)$ is the false alarm probability, which represents the likelihood of making decision $\mathcal{D}_1$ when hypothesis $\mathcal{H}_0$ is true, similarly, $P_{\sf{MD}} = \Pr(\mathcal{D}_0|\mathcal{H}_1)$ is the missed detection probability,  which represents the likelihood of making decision $\mathcal{D}_0$ when hypothesis $\mathcal{H}_1$ is true. Specifically, $\xi=0$ implies that Willie can distinguish both hypotheses perfectly, while $\xi=1$ implies that Willie fails to make a correct detection and behaves like a random guess. 

\subsection{Analysis of DEP at Willie\label{subsec:ADW}}
In this subsection, the analytical expressions of the false alarm and missed detection probabilities at Willie are first derived, from which the optimal detection threshold $\Gamma^\star$ is obtained by employing the first-order optimality condition. We consider the worst-case scenario where Willie can acquire the $\Gamma^\star$ that minimizes the DEP~\cite{wang2021intelligent,si2021covert,lv2021covert}. 



To be specific, the false alarm probability is calculated as
\begin{equation}
    {P_{{\sf{FA}}}} = \Pr \left( {{T_{\sf{W}}} \geqslant \Gamma |\mathcal{H}{_0}} \right) = \Pr \left( {{M_{\sf{W}}}\sigma _{\sf{W}}^2 \geqslant \Gamma |\mathcal{H}{_0}} \right).
\end{equation}
Similarly, define the received power at Willie as $z = \left\| {{\mathbf{F\Theta G}}{{\mathbf{W}}_{{\sf{RF}}}}{{\mathbf{W}}_{{\sf{BB}}}}} \right\|_{\sf{F}}^2$ , the missed detection probability is thus
\begin{equation}
    \begin{aligned}
        {P_{{\sf{MD}}}} &= \Pr \left( {{T_{\sf{W}}} \leqslant \Gamma |\mathcal{H}{_1}} \right) \\ 
   &= \Pr \left( {z + {M_{\sf{W}}}\sigma _{\sf{W}}^2 \leqslant \Gamma |\mathcal{H}{_1}} \right). \\ 
    \end{aligned}
\end{equation}
Then, the DEP for arbitrary $\Gamma$ is given by
\begin{equation} \label{eq:DEP}
    \begin{aligned}
  \xi \left( \Gamma  \right) 
   &= \Pr \left( {{M_{\sf{W}}}\sigma _{\sf{W}}^2 \geqslant \Gamma } \right) + \Pr \left( {z + {M_{\sf{W}}}\sigma _{\sf{W}}^2 \leqslant \Gamma } \right) \\ 
   &= 1 - \Pr \left( {\frac{{\Gamma  - z}}{{{M_{\sf{W}}}}} < \sigma _{\sf{W}}^2 < \frac{\Gamma }{{{M_{\sf{W}}}}}} \right) \\ 
   &= 1 - \int_{\max \left( {\frac{{\Gamma  - z}}{{{M_{\sf{W}}}}},\frac{1}{\rho }\widehat \sigma _{\sf{W}}^2} \right)}^{\min \left( {\frac{\Gamma }{{{M_{\sf{W}}}}},\rho \widehat \sigma _{\sf{W}}^2} \right)} {{f_{\sigma _{\sf{W}}^2}}\left( x \right)dx}.  \\ 
    \end{aligned} 
\end{equation}
Clearly, $\min \left( {\frac{\Gamma }{{{M_{\sf{W}}}}},\rho \widehat \sigma _{\sf{W}}^2} \right) = \frac{\Gamma }{{{M_{\sf{W}}}}}$, which is proved by contradiction as follows. We assume that $\frac{\Gamma }{{{M_{\sf{W}}}}} >\rho \widehat  \sigma _{\sf{W}}^2$, i.e., $\Gamma>M_{\sf W}\rho \widehat \sigma _{\sf{W}}^2$. However, according to the PDF of the bounded uncertainty model in~\eqref{eq:np_pdf} and the test statistic~\eqref{eq:RP_hyp}, the maximum averaged received power under hypothesis $\mathcal{H}_0$ is $M_{\sf W}\rho \widehat \sigma _{\sf{W}}^2$. In covert communications, it is necessary to control the leakage power ${\left\| {{\mathbf{F\Theta G}}{{\mathbf{W}}_{{\sf{RF}}}}{{\mathbf{W}}_{{\sf{BB}}}}} \right\|_{\sf{F}}^2}$ within a certain range to ensure the probability distributions of $T_{\sf W}$ under two hypotheses overlap with each other. Otherwise, the DEP $\xi$ equals zero. Therefore, it follows that $\Gamma^\star$ cannot exceed $M_{\sf W}\rho \widehat \sigma _{\sf{W}}^2$, i.e., $\frac{\Gamma }{{{M_{\sf{W}}}}} <\rho \widehat  \sigma _{\sf{W}}^2$. 

With the PDF of noise uncertainty model given in~\eqref{eq:np_pdf},  taking the differentiation~\eqref{eq:DEP} with respect to $\Gamma$ yields
\begin{equation}
    \frac{{\partial \xi \left( \Gamma  \right)}}{{\partial \Gamma }} = \left\{ {\begin{array}{*{20}{l}}
  { - \frac{{{M_{\sf{W}}}}}{{2\ln \left( \rho  \right)\Gamma }},}&{{\text{if  }}\Gamma  \leqslant \frac{{{M_{\sf{W}}}}}{\rho }\widehat \sigma _{\sf{W}}^2 + z}, \\ 
  { - \frac{1}{{2\ln \left( \rho  \right)}}\left( {\frac{{{M_{\sf{W}}}}}{\Gamma } - \frac{{{M_{\sf{W}}}}}{{\Gamma  - z}}} \right),}&{{\text{if  }}\Gamma  > \frac{{{M_{\sf{W}}}}}{\rho }\widehat \sigma _{\text{W}}^2 + z}. 
\end{array}} \right.
\end{equation}
It follows that  $\frac{{\partial \xi \left( \Gamma  \right)}}{{\partial \Gamma }}\leqslant 0$ when $\Gamma  \leqslant \frac{{{M_{\sf{W}}}}}{\rho }\widehat \sigma _{\sf{W}}^2 + z$, and  $\frac{{\partial \xi \left( \Gamma  \right)}}{{\partial \Gamma }}> 0$ when   $\Gamma  > \frac{{{M_{\sf{W}}}}}{\rho }\widehat \sigma _{\text{W}}^2 + z$. Hence, the optimal detection threshold $\Gamma^\star$ is derived as
\begin{equation}
    \Gamma^\star = \min \left( {\frac{{{M_{\sf{W}}}}}{\rho }\widehat \sigma _{\sf{W}}^2 + z,\rho {M_{\sf{W}}}\widehat \sigma _{\sf{W}}^2} \right).
\end{equation}
Substituting $\Gamma^\star$ into \eqref{eq:DEP} yields the minimum DEP $\xi^\star$ at Willie
\begin{equation}
    \begin{aligned}
  & {\xi ^\star} = 1 - \int_{\frac{1}{\rho }\widehat \sigma _{\sf{W}}^2}^{{\Gamma ^ \star }} {{f_{\sigma _{\sf{W}}^2}}\left( x \right)dx}  \\ 
   &= \left\{ {\begin{array}{*{20}{l}}
  {1 - \frac{1}{{2\ln \left( \rho  \right)}}\ln \left( {1 + \frac{{\rho z}}{{{M_{\sf{W}}}\widehat \sigma _{\sf{W}}^2}}} \right),}&{{\text{if  }}z \leqslant {M_{\sf{W}}}\widehat \sigma _{\sf{W}}^2\left( {\rho  - \frac{1}{\rho }} \right)}, \\ 
  {0,}&{{\text{if  }}z > {M_{\sf{W}}}\widehat \sigma _{\sf{W}}^2\left( {\rho  - \frac{1}{\rho }} \right)}. 
\end{array}} \right. \\ 
\end{aligned} 
\end{equation}
When $z > {M_{\sf{W}}}\widehat \sigma _{\sf{W}}^2\left( {\rho  - \frac{1}{\rho }} \right)$, the DEP equals zeros, which implies Willie can distinguish the two hypotheses perfectly. Hence, in this paper, we focus on the following scenario
\begin{equation}\label{eq:z_range1}
  z \leqslant {M_{\sf{W}}}\widehat \sigma _{\sf{W}}^2\left( {\rho  - \frac{1}{\rho }} \right).  
\end{equation}
Given a predefined small value $\kappa\in [0,1]$ that specifies the covertness level, the covertness requirement can be given as $\xi\geqslant 1-\kappa$. Hence, $z$ should satisfy the following condition
\begin{equation}\label{eq:z_range2}
    z \leqslant \frac{{\left( {{e^{2\kappa \ln \left( \rho  \right)}} - 1} \right){M_{\text{W}}}\widehat \sigma _{\text{W}}^2}}{\rho }.
\end{equation}
From \eqref{eq:z_range1} and \eqref{eq:z_range2}, under the covertness requirement, the maximally allowed leakage power at Willie is obtained as 
\begin{equation}\label{eq:pleak}
    p_{\sf leak}=\min \left( {{M_{\sf{W}}}\widehat \sigma _{\sf{W}}^2\left( {\rho  - \frac{1}{\rho }} \right),\frac{{\left( {{e^{2\kappa \ln \left( \rho  \right)}} - 1} \right){M_{\sf{W}}}\widehat \sigma _{\sf{W}}^2}}{\rho }} \right).
\end{equation}
It is worth noting that the value range of $\kappa$ is typically small, being less than 0.5. According to the bounded uncertainty model in \eqref{eq:np_pdf}, the degree of noise uncertainty $\rho$ is greater than 1. It can be proved that the first term of \eqref{eq:pleak} is always greater than the second term. This indicates that the covertness requirement will act as an active constraint on the leakage power. According to the given conditions, it can be observed that within a certain range, $p_{\sf leak}$ is a monotonically increasing function with respect to $\kappa$. In other words, when the level of covert requirement is higher, the power that can be leaked to Willie will decrease. Hence, the covertness constraint for system design is obtained as follows
\begin{equation}\label{eq:covert_req}
    {\left\| {{\mathbf{F\Theta G}}{{\mathbf{W}}_{{\sf{RF}}}}{{\mathbf{W}}_{{\sf{BB}}}}} \right\|_{\sf{F}}^2} \leqslant p_{\sf leak}.
\end{equation}

\section{Joint Beamforming Design for Covert Communication}\label{Sub:cov_comm_deign}
In this section, an optimization problem is first formulated to maximize the achievable covert rate in the near-field region. Then, an efficient AO-based algorithm is designed to solve the formulated problem. Due to the extremely large-scale antenna array size of the XL-RIS, traditional semidefinite relaxation (SDR)-based optimization method~\cite{luo2010semidefinite}, which requires enormous computational cost, is no longer applicable for XL-RIS's phase coefficient optimization. Therefore, a low-complexity ADMM-based algorithm is proposed for XL-RIS optimization. Finally, the convergence and computational complexity are analyzed for the proposed algorithm.

\subsection{Problem Formulation}
In this paper, we aim to maximize the achievable covert rate by jointly optimizing Alice's hybrid beamforming and XL-RIS's reflection coefficient matrix. Based on the achievable covert rate given in~\eqref{eq:capacity} and covert constraint derived in~\eqref{eq:covert_req}, the optimization problem can be formulated as 

\begin{maxi!}[2]
{{{\bf{W}}_{{\sf{RF}}}}{\bf{,}}{{\bf{W}}_{{\sf{BB}}}},{\bf{\Theta }}} {C_{\sf{B}}\left( {{\bf{W}}_{{\sf{RF}}}}{\bf{,}}{{\bf{W}}_{{\sf{BB}}}},{\bf{\Theta }}\right)}{\label{eq:overall_pro}}{}
\addConstraint{\left\| {{{\bf{W}}_{{\sf{RF}}}}{{\bf{W}}_{{\sf{BB}}}}} \right\|_{\sf{F}}^2 }{\leqslant {P_{\max }}\label{eq:OP_MPC}}
\addConstraint{\left\| {{\bf{F\Theta G}}{{\bf{W}}_{{\sf{RF}}}}{{\bf{W}}_{{\sf{BB}}}}} \right\|_{\sf{F}}^2 }{\leqslant p_{\sf leak} \label{eq:OP_CRC}}
\addConstraint{\left| {{\theta _n}} \right| = 1,\forall n = 1, \ldots ,N}{\label{eq:OP_RISC}}
\addConstraint{w_{{\sf{RF}}}^{\left( {i,j} \right)} \in {\cal F}, }{\forall i=1,\ldots,M_{\sf A}, \forall j=1,\ldots {M_{\sf{A}}^{\sf RF}\label{eq:OP_ABC}}} 
\end{maxi!} 
where \eqref{eq:OP_MPC} is Alice's total power constraint with the maximum transmit power $P_{\max}$, \eqref{eq:OP_CRC} is the covertness requirement constraint with $p_{\sf leak}$ given in~\eqref{eq:pleak}, \eqref{eq:OP_RISC} is the unit-modulus constraints for XL-RIS's reflection coefficients, and \eqref{eq:OP_ABC} is unit-modulus constraints for Alice's analog beamformer.

The formulated problem is challenging to solve due to the highly coupled optimization variables $\mathbf{W}_{\sf RF}$, $\mathbf{W}_{\sf BB}$ and $\bf{\Theta}$ in the objective function and the constraints~\eqref{eq:OP_MPC} and~\eqref{eq:OP_CRC}, as well as the large-scale intractable unit-modulus constraints~\eqref{eq:OP_RISC}. To solve this problem, we propose to optimize the hybrid beamforming at Alice and the reflection coefficient matrix at XL-RIS alternatively until the process is converged. 


\subsection{Hybrid Beamforming Design at Alice}
Let ${{\bf{H}}_{\sf{B}}} = {\bf{H\Theta G}}$ and ${{\bf{H}}_{\sf{W}}} = {\bf{F\Theta G}}$ denote the equivalent channel from Alice to Bob and from Alice to Willie, respectively. For given reflection coefficient matrix $\bf{\Theta}$, the sub-problem of designing hybrid beamforming at Alice can be formulated as follows
\begin{maxi!}[2]
    {{{\bf{W}}_{{\rm{RF}}}},{{\bf{W}}_{{\rm{BB}}}}}{\log_2 \left|  {{{\bf{I}}_{{M_{\sf{B}}}}} \!+\! \sigma _{\sf{B}}^{ - 2}{{\bf{H}}_{\sf{B}}}{{\bf{W}}_{{\sf{RF}}}}{{\bf{W}}_{{\sf{BB}}}}{\bf{W}}_{{\rm{BB}}}^{\sf H}{{\bf{W}}_{{\sf{RF}}}^{\sf H}}{\bf{H}}_{\sf{B}}^{\sf H}}  \right|}{\label{eq:HB_beamforing_subp}}{}
    \addConstraint{\left\| {{{\bf{H}}_{\sf{W}}}{{\bf{W}}_{{\rm{RF}}}}{{\bf{W}}_{{\rm{BB}}}}} \right\|_{\rm{F}}^2 \leqslant p_{\sf leak}}{}
    \addConstraint{\eqref{eq:OP_MPC}, \eqref{eq:OP_ABC}.}{}
\end{maxi!}

To solve the subproblem~\eqref{eq:HB_beamforing_subp}, a two-stage algorithm is proposed for designing the hybrid beamformers. Specifically, the widely adopted WMMSE algorithm is first utilized to design the fully digital beamformer. Then, the digital baseband precoder is derived in closed form, and the analog phase shifts precoder is optimized by the MO algorithm. 
\subsubsection{Stage I: Fully-Digital Beamformer Design} The fully digital beamforming is first designed, which provides an upper bound on the rate performance for the considered hybrid beamforming system. Accordingly, the hybrid beamforming subproblem ~\eqref{eq:HB_beamforing_subp} is reformulated as 
\begin{maxi!}[2]
    {{\mathbf{W}_{\sf FD}}}{ \log_2  \left| {{{\mathbf{I}}_{{M_{\sf{B}}}}} + \sigma _{\sf{B}}^{ - 2}{{\mathbf{H}}_{\sf{B}}}{\mathbf{W}_{\sf FD}}{{\mathbf{W}_{\sf FD}^{\sf H}}}{\mathbf{H}}_{\sf{B}}^{\sf H}} \right|}{\label{eq:stageI}}{}
   \addConstraint{\left\| {\mathbf{W}_{\sf FD}} \right\|_{\sf{F}}^2 \leqslant {P_{\max }}}{\label{eq:stageI_PC}}
   \addConstraint{\left\| {{{\mathbf{H}}_{\sf{W}}}{\mathbf{W}_{\sf FD}}} \right\|_{\text{F}}^2 \leqslant p_{\sf leak}.}{\label{eq:stageI_CC}}
\end{maxi!}

A WMMSE-based algorithm is designed to solve this subproblem. Its main idea is to transform the original problem into a more tractable form by leveraging the equivalence between the rate maximization problem and the mean-square-error (MSE) minimization problem~\cite{shi2011iteratively}. 

Specifically, the signal vector $\widetilde{\mathbf{s}}$ at Bob is estimated by an introduced linear receive beamforming matrix $\mathbf{U}\in\mathbb{C}^{M_{\sf B}\times L}$ as $\widetilde{\mathbf{s}} = \mathbf{U}^{\sf H}\mathbf{y}_{\sf B}$. 
Then, the MSE matrix at Bob can be written as 
\begin{equation} \label{eq:MSE}
\begin{aligned}
 & \mathbf{E} = {\mathbb{E}_{{\mathbf{s}},{{\mathbf{n}}_{\sf{B}}}}}\left[ {\left( {{{\widetilde {\mathbf{s}}}} - {{\mathbf{s}}}} \right){{\left( {{{\widetilde {\mathbf{s}}}} - {{\mathbf{s}}}} \right)}^{\sf H}}} \right] \\ 
  & = \left( {{{\mathbf{I}}_{{M_{\sf{B}}}}} - {{\mathbf{U}}^{\sf{H}}}{{\mathbf{H}}_{\sf{B}}}{\mathbf{W}_{\sf FD}}} \right){\left( {{{\mathbf{I}}_{{M_{\sf{B}}}}} - {{\mathbf{U}}^{\sf{H}}}{{\mathbf{H}}_{\sf{B}}}{\mathbf{W}_{\sf FD}}} \right)^{\sf H}} + \sigma _{\sf{B}}^2{{\mathbf{U}}^{\sf{H}}}{\mathbf{U}}. \\ 
    \end{aligned} 
\end{equation} 
By introducing a weight matrix ${\mathbf{\Psi }} \succcurlyeq {\mathbf{0}}$ for Bob, the subproblem~\eqref{eq:stageI} can be equivalently reformulated as~\cite[Thm. 1]{shi2011iteratively}
\begin{subequations} \label{eq:stageI_trans}
    \begin{align}
     &  \mathop {\min }\limits_{{\mathbf{\Psi }},{\mathbf{U}},{\mathbf{W}_{\sf FD}}} {\text{  Tr}}\left( {{\mathbf{\Psi E}}} \right) - \log_2  \left| {\mathbf{\Psi }} \right| \label{eq:stageI_trans_OB} \\
    & \quad  \text{s.t.}\quad\quad \eqref{eq:stageI_PC}, \eqref{eq:stageI_CC}. 
    \end{align}
\end{subequations}
Although the transformed problem has more optimization variables than~\eqref{eq:stageI}, the objective function in~\eqref{eq:stageI_trans} is more tractable. The receive beamforming matrix $\mathbf{U}$ and the weight matrix $\bf{\Psi}$ only appear in the objective function~\eqref{eq:stageI_trans_OB}. By setting the derivatives of ~\eqref{eq:stageI_trans_OB} with respective to $\mathbf{U}$ and  $\bf{\Psi}$ to zero, respectively, the optimal solutions can be obtained as 
\begin{equation} \label{eq:U_E}
\begin{aligned}
     \mathbf{U}^\star &= \left( {{\mathbf{H}_{\sf B}}\mathbf{W}_{\sf FD}\mathbf{W}_{\sf FD}^{\sf H}{\mathbf{H}_{\sf B}^{\sf H}} + \sigma _{\sf{B}}^2{{\mathbf{I}}_{{M_{\sf{B}}}}}} \right)^{-1}{{\mathbf{H}}_{\sf{B}}}{\mathbf{W}_{\sf FD}}, \\
     \bf{\Psi}^\star & = \mathbf{E}^{-1}.
\end{aligned}
\end{equation}
Substituting the optimal $\mathbf{U}^\star$ in~\eqref{eq:U_E} into~\eqref{eq:MSE} yields the optimal MSE matrix as follows 
\begin{equation}
\begin{aligned}
      & \mathbf{E}^\star  = \\ & {{\mathbf{I}}_{{M_{\sf{B}}}}} - {\mathbf{W}_{\sf FD}^{\sf H}}\mathbf{H}_{\sf{B}}^{\sf H}\left( {{\mathbf{H}_{\sf B}}\mathbf{W}_{\sf FD}\mathbf{W}_{\sf FD}^{\sf H}{\mathbf{H}_{\sf B}^{\sf H}} + \sigma _{\sf{B}}^2{{\mathbf{I}}_{{M_{\sf{B}}}}}} \right)^{-1}{{\mathbf{H}}_{\sf{B}}}{\mathbf{W}_{\sf FD}}.
\end{aligned}
\end{equation}
 Substituting~\eqref{eq:MSE} into the objective function of \eqref{eq:stageI_trans} and discarding the constant terms, the problem that updates the full-digital beamforming matrix $\mathbf{W}_{\sf FD}$ is transformed as 
\begin{mini!} 
    {{\mathbf{W}_{\sf FD}}}{{\text{Tr}}\left( {{\mathbf{W}_{\sf FD}^{\sf H}}\mathbf{H}_{\sf{B}}^{\sf H}{\mathbf{U\Psi }}{{\mathbf{U}}^{\sf{H}}}{{\mathbf{H}}_{\sf{B}}}{\mathbf{W}_{\sf FD}}} \right) - {\text{Tr}}\left( {\mathbf{\Psi }}{\mathbf{W}_{\sf FD}^{\sf H}\mathbf{H}_{\sf{B}}^{\sf H}{\mathbf{U}}} \right)\nonumber}{\label{eq:WFD_subp}}{}
    \breakObjective{- {\text{Tr}}\left( {{\mathbf{\Psi }}{{\mathbf{U}}^{\sf{H}}}{{\mathbf{H}}_{\sf{B}}}{\mathbf{W}_{\sf FD}}} \right)}
    \addConstraint{\eqref{eq:stageI_PC}, \eqref{eq:stageI_CC}.}{}
\end{mini!}
Problem \eqref{eq:WFD_subp} is a convex optimization problem with respect to $\mathbf{W}_{\sf FD}$, which can be solved by the convex solver packages, such as CVX~\cite{cvx}. However, this method incurs high computational complexity. To mitigate this issue, the Lagrangian dual method is leveraged to derive a  semi-closed-form solution as follows. 

Given that Slater's condition holds for problem~\eqref{eq:WFD_subp}, the Karush-Kuhn-Tucker (KKT) conditions serve as both necessary and sufficient for optimality. By introducing the Lagrange multiplier $\mu\geqslant 0$ associated with the maximum power constraint~\eqref{eq:stageI_PC}, the partial Lagrangian function for problem~\eqref{eq:WFD_subp} is given by
\begin{equation}
    \begin{aligned}
 \mathcal{L}\left( \mathbf{W}_{\sf FD}, \mu \right) &={\text{Tr}}\left( {{\mathbf{W}_{\sf FD}^{\sf H}}\mathbf{H}_{\sf{B}}^{\sf H}{\mathbf{U\Psi }}{{\mathbf{U}}^{\sf{H}}}{{\mathbf{H}}_{\sf{B}}}{\mathbf{W}_{\sf FD}}} \right)  \\ 
 &- {\text{Tr}}\left( {\mathbf{\Psi }}{\mathbf{W}_{\sf FD}^{\sf H}\mathbf{H}_{\sf{B}}^{\sf H}{\mathbf{U}}} \right)  - {\text{Tr}}\left( {{\mathbf{\Psi }}{{\mathbf{U}}^{\sf{H}}}{{\mathbf{H}}_{\sf{B}}}{\mathbf{W}_{\sf FD}}} \right) 
 \\ & + \mu \left( {\left\| {\mathbf{W}_{\sf FD}} \right\|_{\sf{F}}^2 - {P_{\max }}} \right). 
\end{aligned} 
\end{equation}
The dual function can be obtained by solving the following problem
\begin{subequations} \label{eq:orgin_problem}
    \begin{align}
       & f\left( \mu  \right) = \mathop {\min }\limits_{\mathbf{W}_{\sf FD}} ~~ \mathcal{L}\left( {{\mathbf{W}_{\sf FD}},\mu } \right) \\
       & \quad\quad\quad ~~ \text{s.t.} ~~ \eqref{eq:stageI_CC}.
    \end{align}
\end{subequations}
Then, the corresponding dual problem is given as 
\begin{subequations} \label{eq:dual_problem}
    \begin{align}
       &  \mathop {\max }\limits_{\mu} \quad f\left( \mu  \right) \\
       &~~ \text{s.t.} \quad \mu \geqslant 0.
    \end{align}
\end{subequations}
To tackle the dual problem~\eqref{eq:dual_problem}, it is necessary to solve the original problem~\eqref{eq:orgin_problem} under a fixed value of $\mu$. Similarly, by introducing the Lagrange multiplier $\upsilon\geqslant 0$ associated with covertness requirement constraint~\eqref{eq:stageI_CC}, the Lagrangian function for problem~\eqref{eq:orgin_problem} is given by
\begin{equation}
\begin{aligned}
        \mathcal{L\left( {{\mathbf{W}_{\sf FD}},\upsilon } \right)} & = {\text{Tr}}\left( {{\mathbf{W}_{\sf FD}^{\sf H}}\left( {\mathbf{H}_{\sf{B}}^{\sf H}{\mathbf{U\Psi }}{{\mathbf{U}^{\sf{H}}}\mathbf{H}_{\sf{B}}} + \mu {\mathbf{I}}} \right){\mathbf{W}_{\sf FD}}} \right) 
        \\ &- {\text{Tr}}\left( {{\mathbf{\Psi }}{{\mathbf{W}_{\sf FD}^{\sf H}}}{\mathbf{H}_{\sf{B}}^{\sf H}}{\mathbf{U}}} \right) - {\text{Tr}}\left( {{\mathbf{\Psi }}{\mathbf{U}^{\sf{H}}}{{\mathbf{H}}_{\sf{B}}}{\mathbf{W}_{\sf FD}}} \right) \\
        & + \upsilon \left( {\left\| {{{\mathbf{H}}_{\sf{W}}}{\mathbf{W}_{\sf FD}}} \right\|_{\sf{F}}^2 - p_{\sf leak} } \right) - \mu {P_{\max }}.
\end{aligned}
\end{equation}
By letting the first-order derivative of $ \mathcal{L\left( {{\mathbf{W}_{\sf FD}},\upsilon } \right)}$ with respect to $\mathbf{W}_{\sf FD}^*$ be equal to $\mathbf{0}$, the optimal solution is obtained according to the first-order optimality condition as follows
\begin{equation} \label{eq:opt_WFD}
    {\mathbf{W}_{\sf FD}^ \star }\left( \upsilon,\bar \mu  \right) = {\left( {{\mathbf{H}}_{\sf{B}}^{\sf H}{\mathbf{U\Psi }}{{\mathbf{U}}^{\sf{H}}}{{\mathbf{H}}_{\sf{B}}} + \bar \mu {{\mathbf{I}}_{{M_{\sf{B}}}}} + \upsilon {\mathbf{H}}_{\sf{W}}^{\sf H}{{\mathbf{H}}_{\sf{W}}}} \right)^\dag }{\mathbf{H}}_{\sf{B}}^{\sf H}{\mathbf{U\Psi }}.
\end{equation}
where $\bar \mu$ is the value of $\mu$ in the last iteration.

For fixed $\mu$, the $\upsilon$ should be chosen to satisfy the following complementary slackness condition for constraint~\eqref{eq:stageI_CC}
\begin{equation}\label{eq:com_sla_upsilon}
   \upsilon \left( {\left\| {{{\mathbf{H}}_{\sf{W}}}{{\mathbf{W}_{\sf FD}^ \star }}\left( \upsilon, \bar \mu  \right)} \right\|_{\sf{F}}^2 - p_{\sf leak} } \right) = 0,
\end{equation}
From~\eqref{eq:com_sla_upsilon}, the optimal $\mathbf{W}_{\sf FD}$ is derived by considering the following two cases. 
\begin{enumerate}
    \item[a)] Case 1: $\upsilon=0$. If the condition $ {\left\| {{{\mathbf{H}}_{\sf{W}}}{{\mathbf{W}_{\sf FD}^ \star } }\left( 0,\bar\mu  \right)} \right\|_{\sf{F}}^2 - p_{\sf leak} }  \leqslant 0$ holds, then the optimal solution is obtained as ${{\mathbf{W}_{\sf FD}^ \star } }\left( 0,\bar\mu  \right)$.
    \item[b)] Case 2: $\upsilon\ne 0$. From~\eqref{eq:opt_WFD}, it can be easily verified that the received power at Willie $\left\| {{{\mathbf{H}}_{\sf{W}}}{{\mathbf{W}_{\sf FD}^ \star }}\left( \upsilon,\bar\mu  \right)} \right\|_{\sf{F}}^2$ is a monotonically decreasing function of $\upsilon$. Hence the bisection search can be used to find the dual variable $\upsilon$ that satisfies the condition~\eqref{eq:com_sla_upsilon}. 
\end{enumerate}

Once the optimal $\upsilon$ is obtained by the above procedure, the bisection search can be used again to find the dual variable $\mu$ that satisfies the following complementary slackness condition for constraint~\eqref{eq:stageI_PC}
\begin{equation}
     \mu \left( {\left\| {\mathbf{W}_{\sf FD}^\star\left(\bar \upsilon, \mu \right)} \right\|_{\sf{F}}^2 - {P_{\max }}} \right) = 0,
\end{equation}
where $\bar \upsilon$ is the value of $\upsilon$ in the last iteration. The steps for finding the optimal dual variables $\mu$ and $\upsilon$ via bisection search is given in Algorithm~\ref{alg1_1}.

Furthermore, the whole procedure for solving problem~\eqref{eq:WFD_subp} is summarized as Algorithm~\ref{alg1_2}. The sequence of $\mathbf{W}_{\sf FD}$ generated by the WMMSE algorithm converges to the KKT point of problem~\eqref{eq:stageI}, which has been verified in~\cite{pan2017joint,shi2015secure}.


\begin{algorithm}[t]
	\renewcommand{\algorithmicrequire}{\textbf{Input:}}
	\renewcommand{\algorithmicensure}{\textbf{Output:}}
	\caption{Bisection Search Method to Solve Problem \eqref{eq:WFD_subp}}
	\label{alg1_1}
	\begin{algorithmic}[1]
		\STATE Initialization: the convergence accuracy $\varepsilon$, the lower bound $\mu_{\sf l}$, $\upsilon_{\sf l}^0$ and upper bound $\mu_{\sf u}$, $\upsilon_{\sf u}^0$ for $\mu$ and $\upsilon$, respectively.
		\REPEAT
		\STATE Update $\bar \mu = (\mu_{\sf l}+\mu_{\sf u})/2$, $\upsilon_{\sf l} = \upsilon_{\sf l}^0$ and $\upsilon_{\sf u} = \upsilon_{\sf u}^0$.
        \IF{$ {\left\| {{{\mathbf{H}}_{\sf{W}}}{{\mathbf{W}_{\sf FD}^ \star } }\left( 0, \bar \mu  \right)} \right\|_{\sf{F}}^2 - p_{\sf leak} }  \leqslant 0$}
           \STATE         $\upsilon = 0$,
        \ELSE
            \REPEAT
            \STATE Update $\upsilon = (\upsilon_{\sf l}+\upsilon_{\sf u})/2$,
            \IF{$ {\left\| {{{\mathbf{H}}_{\sf{W}}}{{\mathbf{W}_{\sf FD}^ \star } }\left(\upsilon , \bar \mu  \right)} \right\|_{\sf{F}}^2 \leqslant  p_{\sf leak} }$}
            \STATE $\upsilon_{\sf l} = \upsilon$,
            \ELSE
             \STATE $\upsilon_{\sf u} = \upsilon$.
             \ENDIF
    		\UNTIL $|\upsilon_{\sf l} - \upsilon_{\sf u}|< \varepsilon$. 
      \STATE Update $\mathbf{W}_{\sf FD}^{(t)}$ in~\eqref{eq:opt_WFD}. 
       \ENDIF
       \IF{$ {\left\| \mathbf{W}_{\sf FD}^{(t)} \right\|_{\sf{F}}^2}  \leqslant {P_{\max }}$} 
       \STATE $\mu_{\sf l} =\bar \mu$,
       \ELSE
       \STATE  $\mu_{\sf u} =\bar  \mu$.
       \ENDIF
      \UNTIL $|\mu_{\sf l} - \mu_{\sf u}|< \varepsilon$, terminal. Otherwise, go to step 5.
	\end{algorithmic}  
\end{algorithm}

\begin{algorithm}[t]
	\renewcommand{\algorithmicrequire}{\textbf{Input:}}
	\renewcommand{\algorithmicensure}{\textbf{Output:}}
	\caption{WMMSE Method for Fully-Digital Beamformer Design}
	\label{alg1_2}
	\begin{algorithmic}[1]
		\STATE Initialization: the convergence accuracy $\varepsilon$, the fully-digital beamfomer $\mathbf{W}_{\sf FD}^{(0)}$,  the maximum number of iterations $t_{\max}$ and iteration index $t=1$,  calculate the objective function of~\eqref{eq:stageI} as $f\left(\mathbf{W}_{\sf FD}^{(0)} \right)$.
        \REPEAT 
        \STATE Update $\mathbf{U}$ and $\mathbf{\Psi}$ in~\eqref{eq:U_E}.
	\STATE	Update $\mathbf{W}_{\sf FD}$ by solving Problem~\eqref{eq:WFD_subp} using Algorithm~\ref{alg1_1}.
      \UNTIL $\frac{\left|f\left(\mathbf{W}_{\sf FD}^{(t)} \right) - f\left(\mathbf{W}_{\sf FD}^{(t-1)} \right)\right|}{f\left(\mathbf{W}_{\sf FD}^{(t-1)} \right)}<\varepsilon$ or $t \geqslant t_{max}$, terminate.
		\ENSURE  $\mathbf{W}_{\sf FD}^\star$
	\end{algorithmic} 
\end{algorithm}

\subsubsection{Stage II: Hybrid Beamformer Design}
Given the optimal fully digital beamformer $\mathbf{W}_{\sf FD}^\star$ obtained in stage I, we further consider the hybrid beamformer design due to the hardware constraints. To ensure that the inequality \eqref{eq:OP_CRC} still holds for hybrid architectures, we assume that the number of RF chains is greater than or equal to twice the number of data streams in this work, so that the beamforming matrices of the hybrid architecture can perfectly realize any fully digital beam matrix~\cite{sohrabi2016hybrid}. When the above assumption is not satisfied, we can preset a small margin $\delta$ so that it satisfies the inequality $\left\| {{{\mathbf{H}}_{\sf{W}}}{\mathbf{W}_{\sf FD}}} \right\|_{\text{F}}^2 \leqslant p_{\sf leak}- \delta$, which is more stringent compared with the inequality~\eqref{eq:stageI_CC}. Hence, the inequality~\eqref{eq:stageI_CC} can be guaranteed even if the fully digital beamforming matrix cannot be approximated perfectly. To approximate the performance of the fully digital beamformer, a hybrid beamformer design problem is formulated to obtain the near-optimal analog beamformer $\mathbf{W}_{\sf RF}$ and baseband digital beamformer $\mathbf{W}_{\sf BB}$ as follows
\begin{mini!}
{{\mathbf{W}_{{\sf{RF}}}},{{\mathbf{W}}_{{\sf{BB}}}}}{\left\| {{{\mathbf{W}}_{{\sf{FD}}}} - {{\mathbf{W}}_{{\sf{RF}}}}{{\mathbf{W}}_{{\sf{BB}}}}} \right\|_{\sf{F}}^2}{\label{eq:HBD}}{}
\addConstraint{\eqref{eq:OP_ABC}.}{}
\end{mini!}
Since the analog beamformer and the baseband digital beamformer are coupled with each other, an alternating minimization method is employed to optimize both beamformers.
\begin{itemize}
    \item \textit{Baseband Digital Beamfomer Design}: Given $\mathbf{W}_{\sf RF}$, the optimization problem~\eqref{eq:HBD} can be expressed as 
    \begin{mini!}
{{{\mathbf{W}}_{{\sf{BB}}}}}{\left\| {{{\mathbf{W}}_{{\sf{FD}}}} - {{\mathbf{W}}_{{\sf{RF}}}}{{\mathbf{W}}_{{\sf{BB}}}}} \right\|_{\sf{F}}^2}{}{}
\end{mini!}
which is a standard least-square problem. Thus, according to the first-order optimality conditions, the optimal $\mathbf{W}_{\sf BB}$ can be derived as
\begin{equation} \label{eq:WBB}
  {\mathbf{W}}_{{\sf{BB}}}^ \star  = {\left( {\mathbf{W}_{{\sf{RF}}}^{\sf H}{{\mathbf{W}}_{{\sf{RF}}}}} \right)^{-1}}\mathbf{W}_{\sf{RF}}^{\sf H}{{\mathbf{W}}_{\sf FD}}.
\end{equation}

\item \textit{Analog Beamformer Design}: Given $\mathbf{W}_{\sf BB}$, the optimization problem~\eqref{eq:HBD} is reduced as 
\begin{mini!}
{{\mathbf{W}_{{\sf{RF}}}}}{\left\| {{{\mathbf{W}}_{{\sf{FD}}}} - {{\mathbf{W}}_{{\sf{RF}}}}{{\mathbf{W}}_{{\sf{BB}}}}} \right\|_{\sf{F}}^2}{\label{eq:ABD}}{}
\addConstraint{\eqref{eq:OP_ABC}.}{}
\end{mini!}
which can be transformed into the equivalently form~\eqref{eq:ELF} shown at the top of the next page.
\begin{figure*}
\begin{mini!}
{{\mathbf{W}_{{\sf{RF}}}}}{e\left( {{\text{vec}}\left( {{{\mathbf{W}}_{{\sf{RF}}}}} \right)} \right) = \left\| {{\text{vec}}\left( {{{\mathbf{W}}_{{\sf{FD}}}}} \right) - \left( {{\mathbf{W}}_{{\sf{BB}}}^{\sf{T}} \otimes {{\mathbf{I}}_{{M_{\sf{A}}^{\sf RF}}}}} \right){\text{vec}}\left( {{{\mathbf{W}}_{{\sf{RF}}}}} \right)} \right\|^2 \label{eq:ELF_ob}}{\label{eq:ELF}}{}
\addConstraint{\eqref{eq:OP_ABC}.}{}
\end{mini!}
\end{figure*}

The objective function $e\left( {{\text{vec}}\left( {{{\mathbf{W}}_{{\sf{RF}}}}} \right)} \right)$ is continuous and differentiable with respect to $ {{\text{vec}}\left( {{{\mathbf{W}}_{{\sf{RF}}}}} \right)} $, and the constraint set $\mathcal{F}$ is a complex circle manifold. Hence, the MO algorithm~\cite{manopt} can be adopted to design the analog beamformer. In general, MO is mainly divided into three steps:
\begin{enumerate}
    \item[a)] \textit{Computation of Riemannian Gradient}: The Riemannian gradient ${\text{grad}}_{{{\text{vec}}\left( {{{\mathbf{W}}_{{\sf{RF}}}}} \right)}} e\left( {{\text{vec}}\left( {{{\mathbf{W}}_{{\sf{RF}}}}} \right)} \right)$, which is the orthogonal projection of the Euclidean gradient $\nabla_{{{\text{vec}}\left( {{{\mathbf{W}}_{{\sf{RF}}}}} \right)}} e\left( {{\text{vec}}\left( {{{\mathbf{W}}_{{\sf{RF}}}}} \right)} \right)$ of~\eqref{eq:ELF_ob} onto the tangent space, is given by
    \begin{equation}
    \begin{aligned}
      & {\text{grad}}_{{{\text{vec}}\left( {{{\mathbf{W}}_{{\sf{RF}}}}} \right)}} e \left( {{\text{vec}}\left( {{{\mathbf{W}}_{{\sf{RF}}}}} \right)} \right)   
     \\ & = \nabla_{{{\text{vec}}\left( {{{\mathbf{W}}_{{\sf{RF}}}}} \right)}} e\left( {{\text{vec}}\left( {{{\mathbf{W}}_{{\sf{RF}}}}} \right)} \right)  
     \\ & - \Re\{\nabla_{{{\text{vec}}\left( {{{\mathbf{W}}_{{\sf{RF}}}}} \right)}} e\left( {{\text{vec}}\left( {{{\mathbf{W}}_{{\sf{RF}}}}} \right)} \right) \circ {{\text{vec}}\left( {{{\mathbf{W}}_{{\sf{RF}}}}} \right)}^* \}  
     \\ & \qquad \circ  {{\text{vec}}\left( {{{\mathbf{W}}_{{\sf{RF}}}}} \right)},
    \end{aligned}
    \end{equation}
    where the $\nabla_{{{\text{vec}}\left( {{{\mathbf{W}}_{{\sf{RF}}}}} \right)}} e\left( {{\text{vec}}\left( {{{\mathbf{W}}_{{\sf{RF}}}}} \right)} \right)$ is given in~\eqref{eq:EG} shown at the top of the next page.

\begin{figure*}
\begin{equation} \label{eq:EG}
  {\nabla _{{\text{vec}}\left( {{{\mathbf{W}}_{{\sf{RF}}}}} \right)}}e\left( {{\text{vec}}\left( {{{\mathbf{W}}_{{\sf{RF}}}}} \right)} \right) = 2{\left( {{\mathbf{W}}_{{\sf{BB}}}^{\sf{T}} \otimes {{\mathbf{I}}_{{M_{\sf{A}}^{\sf RF}}}}} \right)^{\sf H}}\left( {{\mathbf{W}}_{{\sf{BB}}}^{\sf{T}} \otimes {{\mathbf{I}}_{{M_{\sf{A}}^{\sf RF}}}}} \right){\text{vec}}\left( {{{\mathbf{W}}_{{\sf{RF}}}}} \right) 
   - 2{\left( {{\mathbf{W}}_{{\sf{BB}}}^{\sf{T}} \otimes {{\mathbf{I}}_{{M_{\sf{A}}^{\sf RF}}}}} \right)^{\sf H}}{\text{vec}}\left( {{{\mathbf{W}}_{{\sf{FD}}}}} \right).
    \end{equation}
{\noindent} \rule[-10pt]{18cm}{0.05em}
\end{figure*}

\item[b)]\textit{ Finding Search Direction}: The search direction, which can be calculated by the tangent vector conjugated to ${\text{grad}}_{{{\text{vec}}\left( {{{\mathbf{W}}_{{\sf{RF}}}}} \right)}} e\left( {{\text{vec}}\left( {{{\mathbf{W}}_{{\sf{RF}}}}} \right)} \right)$, is given by
\begin{equation}
    \mathbf{d} = -{\text{grad}}_{{{\text{vec}}\left( {{{\mathbf{W}}_{{\sf{RF}}}}} \right)}} e\left( {{\text{vec}}\left( {{{\mathbf{W}}_{{\sf{RF}}}}} \right)} \right)+ \varpi \mathbb{T}\left( \bar{\mathbf{d}} \right), 
\end{equation}
where $\varpi$ is the conjugate gradient update parameter, $\bar{\mathbf{d}}$ is the search direction in the last iteration, and the vector transport function $
   \mathbb{T}\left( \mathbf{d} \right) = \bar{\mathbf{d}} - \Re\{ \mathbf{d} \circ {{\text{vec}}\left( {{{\mathbf{W}}_{{\sf{RF}}}}} \right)}^*  \} \circ {{\text{vec}}\left( {{{\mathbf{W}}_{{\sf{RF}}}}} \right)}$.
   
\item[c)] \textit{Retraction}: The process of retraction is to find the next solution by mapping the current point ${{\text{vec}}\left( {{{\mathbf{W}}_{{\sf{RF}}}}} \right)}$ on the tangent space back to the complex circle manifold. Mathematically, it is given by
\begin{equation}
    {{\text{vec}}\left( {{\mathbf{W}_{{\sf{RF}}}}} \right)}_n = \frac{\left( {{\text{vec}}\left( {{\mathbf{ \bar W}}_{{\sf{RF}}}}\right) + \tau\mathbf{d} } \right)_n}{\left| {{\text{vec}}\left( {{\mathbf{\bar W}_{{\sf{RF}}}}}\right) + \tau\mathbf{d}}  \right|_n},
\end{equation}
where ${\text{vec}}\left( {{\mathbf{\bar W}_{{\sf{RF}}}}}\right)$ is the solution in the last iteration, and $\tau$ is the step size updated by Armijo rule generally. 
\end{enumerate} 
\end{itemize}

The steps for hybrid beamformer design is summarized as Algorithm~\ref{alg2}, and ${\mathbf{\bar W}}_{{\sf{BB}}}$ is the solution in the last iteration.
\begin{algorithm}[t]
	\renewcommand{\algorithmicrequire}{\textbf{Input:}}
	\renewcommand{\algorithmicensure}{\textbf{Output:}}
	\caption{Alternating Minimization Method for Hybrid Beamformer Design}
	\label{alg2}
	\begin{algorithmic}[1]
		\STATE \algorithmicrequire ~ the fully-digital beamfomer $\mathbf{W}_{\sf FD}$.
            \STATE Initialization: the convergence accuracy $\varepsilon$, and iteration index $t=1$.
		\REPEAT
		\STATE Update $\mathbf{W}_{\sf BB}$ in~\eqref{eq:WBB},
            \STATE Update $\mathbf{W}_{\sf RF}$ by solving problem~\eqref{eq:ABD} by MO algorithm. 
      \UNTIL $\frac{\left\| {{{\mathbf{W}}_{{\sf{FD}}}^{(t)}} - {{\mathbf{W}}_{{\sf{RF}}}^{(t)}}{{\mathbf{W}}_{{\sf{BB}}}^{(t)}}} \right\|_{\sf{F}}^2}{\left\| {{{\mathbf{W}}_{{\sf{FD}}}^{(t-1)}} - {{\mathbf{W}}_{{\sf{RF}}}^{(t-1)}}{{\mathbf{W}}_{{\sf{BB}}}^{(t-1)}}} \right\|_{\sf{F}}^2} < \varepsilon$ or $t\geqslant t_{\max}$, terminal.
		\ENSURE  $\mathbf{W}_{\sf RF}^\star$, $\mathbf{W}_{\sf BB}^\star$
	\end{algorithmic}  
\end{algorithm}

\subsection{Reflection Coefficient Design at XL-RIS}
For fixed $\mathbf{W}_{\sf RF}$ and $\mathbf{W}_{\sf BB}$, define $\mathbf{\widehat W}_{\sf FD} \triangleq \mathbf{W}_{\sf RF}\mathbf{W}_{\sf BB}$. By substituting the MSE matrix $\mathbf{E}$ in~\eqref{eq:MSE}, ${{\bf{H}}_{\sf{B}}} = {\bf{H\Theta G}}$ and ${{\bf{H}}_{\sf{W}}} = {\bf{F\Theta G}}$ into~\eqref{eq:stageI_trans}, and discarding the terms that are independent of $\bf{\Theta}$, the reflection coefficient design problem can be reformulated as 
\begin{mini!}
{\mathbf{\Theta }}{ {\text{ Tr}}\left( {{{\mathbf{\Theta }}^{\sf{H}}}{{\mathbf{H}}^{\sf{H}}}{\mathbf{U\Psi }}{{\mathbf{U}}^{\sf{H}}}{\mathbf{H\Theta G}}\mathbf{\widehat W}_{\sf FD} \mathbf{\widehat W}_{\sf FD}^{\sf H}{{\mathbf{G}}^{\sf{H}}}} \right) \nonumber}{\label{eq:RCO}}{}
\breakObjective{- {\text{Tr}}\left( {{{\mathbf{\Theta }}^{\sf{H}}}{{\mathbf{H}}^{\sf{H}}}{\mathbf{U\Psi }}\mathbf{\widehat W}_{\sf FD}^{\sf H}{{\mathbf{G}}^{\sf{H}}}} \right)}
\addConstraint{{\text{Tr}}\left( {{{\mathbf{\Theta }}^{\sf{H}}}{{\mathbf{F}}^{\sf{H}}}{\mathbf{F\Theta G}}\mathbf{\widehat W}_{\sf FD} \mathbf{\widehat W}_{\sf FD}^{\sf H}{{\mathbf{G}}^{\sf{H}}}} \right) \leqslant p_{\sf leak}}{}
\addConstraint{\eqref{eq:OP_RISC}.}{}
\end{mini!}

For ease of presentation, define ${\mathbf{ A}} \triangleq {{\mathbf{H}}^{\sf{H}}}{\mathbf{U\Psi }}{{\mathbf{U}}^{\sf{H}}}{\mathbf{H}}$, ${\mathbf{{B}}} \triangleq {\mathbf{G}}\mathbf{\widehat W}_{\sf FD} \mathbf{\widehat W}_{\sf FD}^{\sf H}{{\mathbf{G}}^{\sf{H}}}$, ${\mathbf{C}} = {{\mathbf{H}}^{\sf{H}}}{\mathbf{U\Psi }}\mathbf{\widehat W}_{\sf FD}^{\sf H}{{\mathbf{G}}^{\sf{H}}}$, and ${\mathbf{\bar F}} = {{\mathbf{F}}^{\sf{H}}}{\mathbf{F}}$. The optimization problem~\eqref{eq:RCO} can be simplified as 
\begin{mini!}
{\mathbf{\Theta }}{{\text{ Tr}}\left( {{{\mathbf{\Theta }}^{\sf{H}}}\mathbf{A}{\mathbf{\Theta }}\mathbf{B}} \right) - {\text{Tr}}\left( \mathbf{\Theta}^{\sf H}\mathbf{C} \right) - {\text{Tr}}\left( \mathbf{C}^{\sf H}\mathbf{\Theta} \right) }{\label{eq:RCO2}}{}
\addConstraint{{\text{Tr}}\left( {{{\mathbf{\Theta }}^{\sf{H}}}\mathbf{\bar F}{\mathbf{\Theta}}\mathbf{B}} \right) \leqslant p_{\sf leak}}{}
\addConstraint{\eqref{eq:OP_RISC}.}{}
\end{mini!}
Since $\mathbf{\Theta}$ is a diagonal matrix, by exploiting the matrix identity in~\cite[Eq. (1.10.6)]{zhang2017matrix}, we have the following equivalence transformation
\begin{equation}
\begin{aligned}
        {\text{Tr}}\left( {{{\mathbf{\Theta }}^{\sf{H}}}{\mathbf{A\Theta {\mathbf B}}}} \right) &= {{\mathbf{v}}^{\sf{H}}}\left( {{\mathbf{A}} \circ {{\mathbf{{ B}}}^{\sf T}}} \right){\mathbf{v}}, \\
  {\text{Tr}}\left( {{{\mathbf{\Theta }}^{\sf{H}}}{\mathbf{C}}} \right) &= {{\mathbf{v}}^{\sf{H}}}{\mathbf{c}},\\
  {\text{Tr}}\left( {{{\mathbf{\Theta }}^{\sf{H}}}{\mathbf{\bar F\Theta { B}}}} \right) &= {{\mathbf{v}}^{\sf{H}}}\left( {{\mathbf{\bar F}} \circ {{\mathbf{{ B}}}^{\sf T}}} \right){\mathbf{v}},
\end{aligned}
\end{equation}
where the vector $\mathbf{c}$ is the diagonal elements of matrix $\mathbf{C}$, and ${\mathbf{v}} = {\left[ {{e^{\jmath{\theta _1}}}, \ldots , {e^{\jmath{\theta _N}}}} \right]^{\sf T}}$.

Define ${\mathbf{\Xi }} \triangleq \mathbf{A} \circ \mathbf{B}$ and $\mathbf{\Upsilon} \triangleq \mathbf{\bar F}\circ \mathbf{B}^{\sf T}$. The optimization problem~\eqref{eq:RCO2} can be further simplified as 
\begin{mini!}[2]
{\mathbf{v }}{{{\mathbf{v}}^{\sf{H}}}{\mathbf{\Xi v}} - 2\Re \left\{ {{{\mathbf{v}}^{\sf{H}}}{\mathbf{c}}} \right\} \label{eq:RCO3_ob}}{\label{eq:RCO3}}{}
\addConstraint{{{\mathbf{v}}^{\sf{H}}}\mathbf{\Upsilon} {\mathbf{v}} \leqslant p_{\sf leak}}{ \label{eq:RCO3_c1}}
\addConstraint{ \eqref{eq:OP_RISC}.}{}
\end{mini!}
It is noted that $\mathbf{A}$ can be verified to be a non-negative semidefinite matrix due to the fact that ${\mathbf{\Psi }} \succcurlyeq {\mathbf{0}}$. And according to the Schur product theorem~\cite[Ch. 6]{zhang2006schur}, $\mathbf{\Xi}$ is a non-negative semidefinite matrix. Similarly, $\mathbf{\Upsilon}$ is also a non-negative semidefinite matrix. Thus, the objective function of~\eqref{eq:RCO3} is a convex quadratic function with respect to $\mathbf{v}$, and \eqref{eq:RCO3_c1} is a convex quadratic constraint. However, the reflection coefficients are subject to unit-modulus constraints that are non-convex. 

In this paper, an ADMM-based algorithm is proposed to solve the non-convex optimization problem~\eqref{eq:RCO3}. The reasons for adopting ADMM are twofold. Firstly, the ADMM algorithm is efficient for solving optimization problems with large-scale variables, as it decomposes the original problem into smaller and more manageable subproblems. Due to the high computation complexity, SDR-based algorithm is not applicable here. Secondly, the AMDD algorithm is capable of handling the intractable unit-modulus constraints efficiently. 

Based on the ADMM algorithm, an auxiliary variable $\pmb{\varphi}=\left[\varphi_1,\dots,\varphi_N\right]^{\sf T}$ is introduced to transform the optimization problem~\eqref{eq:RCO3} as follows
\begin{mini!}
{\mathbf{v },\pmb{\varphi}}{{{\mathbf{v}}^{\sf{H}}}{\mathbf{\Xi v}} - 2\Re \left\{ {{{\mathbf{v}}^{\sf{H}}}{\mathbf{c}}} \right\}}{}{}
\addConstraint{{{\mathbf{v}}^{\sf{H}}}\mathbf{\Upsilon} {\mathbf{v}} \leqslant p_{\sf leak}}{}
\addConstraint{\left|\theta_n \right|\leqslant 1, \forall n=1,\dots N}{}
\addConstraint{\mathbf{v} = \pmb{\varphi}}{}
\addConstraint{\left|\varphi_n \right| = 1, \forall n=1,\dots N.}{}
\end{mini!}
Then, the augmented Lagrangian function is constructed as
\begin{equation}
    \mathcal{L}_\varrho(\mathbf{v },\pmb{\varphi},\pmb{\xi}) = {{\mathbf{v}}^{\sf{H}}}{\mathbf{\Xi v}} - 2\Re \left\{ {{{\mathbf{v}}^{\sf{H}}}{\mathbf{c}}} \right\} + \frac{\varrho }{2}{\left\| {{\mathbf{v}} - {\pmb{\varphi }} + {\pmb{\xi }}/\varrho } \right\|^2},
\end{equation}
where $\varrho>0$ is a predefined penalty parameter, and $\pmb{\xi}=\left[\xi_{1},\dots, \xi_{N} \right]^{\sf T}$ denotes the scaled dual variable vector for the constraint $\mathbf{v} = \pmb{\varphi}$. Then, the optimization problem is reformulated as
\begin{mini!}[2]
{\mathbf{v },\pmb{\varphi},\pmb{\xi}}{\mathcal{L}_\varrho(\mathbf{v },\pmb{\varphi},\pmb{\xi})}{}{}
\addConstraint{{{\mathbf{v}}^{\sf{H}}}\mathbf{\Upsilon} {\mathbf{v}} \leqslant p_{\sf leak}}{\label{eq:theta_opt_c1}}
\addConstraint{\left|\theta_n \right|\leqslant 1, \forall n=1,\dots N}{\label{eq:theta_opt_c2}}
\addConstraint{\left|\varphi_n \right| = 1, \forall n=1,\dots N.}{\label{eq:theta_opt_c3}}
\end{mini!}

As such, the optimization variables can update iteratively. 
\begin{enumerate}
    \item \textit{Update $\mathbf{v}$}: When $\pmb{\varphi}$ and $\pmb{\xi}$ are fixed, the subproblem for updating $\mathbf{v}$ is expressed as
    \begin{mini!}
        {\mathbf{v }}{\mathcal{L}_\varrho(\mathbf{v },\pmb{\varphi},\pmb{\xi})}{\label{eq:update_V}}{}
        \addConstraint{\eqref{eq:theta_opt_c1},\eqref{eq:theta_opt_c2}.}{}
    \end{mini!}
The problem~\eqref{eq:update_V} is a convex quadratically constrained quadratic program  (QCQP), as both $\bf{\Xi}$ and $\mathbf{\Upsilon}$ are positive semidefinite. Thus it can be solved by existing convex optimization toolbox.    
\item \textit{Update $\pmb{\varphi}$}: When $\mathbf{v}$ and $\pmb{\xi}$ are fixed, the subproblem for updating $\pmb{\varphi}$ is expressed as
    \begin{mini!}
    {\pmb{\varphi }}{ {\left\| {{\mathbf{v}} - {\pmb{\varphi }} + {\pmb{\xi }}/\varrho } \right\|^2}}{}{}
    \addConstraint{\eqref{eq:theta_opt_c3}.}{}
    \end{mini!}
The optimal $\pmb{\varphi}^\star$ can be easily derived by the phase alignment as follows
\begin{equation} \label{eq:update_varphi}
    \pmb{\varphi}^\star = e^{\jmath \left( \varrho\mathbf{v} + \pmb{\xi}\right)}.
\end{equation}
\item Update $\pmb{\xi}$: When $\mathbf{v}$ and $\pmb{\varphi}$ are fixed, the subproblem for updating the dual variable $\pmb{\xi}$ is expressed as
\begin{mini!}
{\pmb{\xi }}{ {\left\| {{\mathbf{v}} - {\pmb{\varphi }} + {\pmb{\xi }}/\varrho } \right\|^2}.}{}{}
\end{mini!}
According to the conjugate gradient descent method, the dual variable $\pmb{\xi}$ is updated by
\begin{equation} \label{eq:update_xi}
    \pmb{\xi} := \pmb{\bar \xi} + 2\varrho\left(\mathbf{v} - \pmb{\varphi}  \right),
\end{equation}
where $\pmb{\bar \xi}$ is the value in the last iteration.
\end{enumerate}

The steps for reflection coefficient design is summarized in Algorithm~\ref{alg3}, and $\mathbf{\bar v}$ is the solution in the last iteration. The ADMM algorithm is guaranteed to converge to the set of stationary solutions~\cite{hong2016convergence}. 
\begin{algorithm}[t]
	\renewcommand{\algorithmicrequire}{\textbf{Input:}}
	\renewcommand{\algorithmicensure}{\textbf{Output:}}
	\caption{ADMM Method for Reflection Coefficient Design}
	\label{alg3}
	\begin{algorithmic}[1]
		\STATE \algorithmicrequire ~ analog beamfomer $\mathbf{W}_{\sf RF}$ and baseband digital beamfomer $\mathbf{W}_{\sf BB}$ .
            \STATE Initialization: convergence accuracy $\varepsilon$, and time iteration index $t=1$.
		\REPEAT
		\STATE Update $\mathbf{v}$ by solving problem~\eqref{eq:update_V},
            \STATE Update $\pmb{\varphi}$ in~\eqref{eq:update_varphi},
            \STATE Update $\pmb{\xi}$ in~\eqref{eq:update_xi}. 
      \UNTIL $\frac{\left\| \mathcal{L}_\varrho(\mathbf{v }^{(t)},\pmb{\varphi}^{(t)},\pmb{\xi}^{(t)}) -\mathcal{L}_\varrho(\mathbf{v}^{(t-1)},\pmb{\varphi}^{(t-1)},\pmb{\xi}^{(t-1)}) \right\|}{\left\| \mathcal{L}_\varrho(\mathbf{v}^{(t-1)},\pmb{\varphi}^{(t-1)},\pmb{\xi}^{(t-1)}) \right\|} < \varepsilon$ or $t \geqslant t_{\max}$, terminal.
		\ENSURE  $\mathbf{v}^\star$
	\end{algorithmic}  
\end{algorithm}

\begin{algorithm}[t]
	\renewcommand{\algorithmicrequire}{\textbf{Input:}}
	\renewcommand{\algorithmicensure}{\textbf{Output:}}
	\caption{Overall Algorithm for XL-RIS empowered Covert Communication Problem~\eqref{eq:overall_pro}}
	\label{alg4}
	\begin{algorithmic}[1]
            \STATE Initialization: convergence accuracy $\varepsilon$,  analog beamfomer $\mathbf{W}_{\sf RF}^{(0)}$ and the baseband digital beamformer $\mathbf{W}_{\sf BB}^{(0)}$, reflection coefficient at the XL-RIS $\mathbf{\Theta}^{(0)}$, maximum number of iterations $t_{\max}$ and iteration index $t=1$, calculate the objective function of~\eqref{eq:overall_pro} as $f\left(\mathbf{W}_{\sf RF}^{(0)} ,\mathbf{W}_{\sf BB}^{(0)},\mathbf{\Theta}^{(0)}\right)$.
		\REPEAT
		\STATE Update $\mathbf{W}_{\sf FD}$ by Algorithm~\ref{alg1_2},
            \STATE Approximate $\mathbf{W}_{\sf RF}$ and $\mathbf{W}_{\sf BB}$ by Algorithm~\ref{alg2},
            \STATE Update $\mathbf{v}$ by Algorithm~\ref{alg3}. 
      \UNTIL $\frac{\left\| f\left(\mathbf{W}_{\sf RF}^{(n)} ,\mathbf{W}_{\sf BB}^{(n)},\mathbf{\Theta}^{(t)}\right) -f\left(\mathbf{W}_{\sf RF}^{(t-1)} ,\mathbf{W}_{\sf BB}^{(t-1)},\mathbf{\Theta}^{(t-1)}\right) \right\|}{\left\| f\left(\mathbf{W}_{\sf RF}^{(t-1)} ,\mathbf{W}_{\sf BB}^{(t-1)},\mathbf{\Theta}^{(t-1)}\right) \right\|} < \varepsilon$ or $t>t_{\max}$, terminal.
		\ENSURE  $\mathbf{W}_{\sf RF}^\star$, $\mathbf{W}_{\sf BB}^\star$, $\mathbf{\Theta}^\star$.
	\end{algorithmic}  
\end{algorithm}

\subsection{Overall Algorithm and Complexity Analysis}
The overall algorithm for solving the overall joint beamforming problem \eqref{eq:overall_pro} is summarized in Algorithm~\ref{alg4}. 

As adopted in the classic work on hybrid beamforming design for mmWave communications~\cite{sohrabi2016hybrid}, we assume that the number of RF chains is greater than or equal to twice the number of data streams~\cite{sohrabi2016hybrid}. Since both the WMMSE and ADMM-based algorithms are guaranteed to converge, the objective value sequence $\left\{f \left(\mathbf{W}_{\sf RF}^{(n)}, \mathbf{W}_{\sf BB}^{(n)},\mathbf{\Theta}^{(n)}\right), t=1,2,\dots\right\}$  of ~\eqref{eq:overall_pro} is guaranteed to converge due to the following inequality

\begin{equation} \label{eq:conv}
\begin{aligned}
  & f \left(\mathbf{W}_{\sf RF}^{(n)} ,\mathbf{W}_{\sf BB}^{(n)},\mathbf{\Theta}^{(n)}\right)  \\ & \quad \quad \leqslant f \left(\mathbf{W}_{\sf RF}^{(n+1)} ,\mathbf{W}_{\sf BB}^{(n+1)},\mathbf{\Theta}^{(n)}\right)  \\  & \quad \quad \leqslant   f \left(\mathbf{W}_{\sf RF}^{(n+1)} ,\mathbf{W}_{\sf BB}^{(n+1)},\mathbf{\Theta}^{(n+1)}\right),   
\end{aligned} 
\end{equation}
which implies the monotonic increase of the generated sequence  $\left\{f \left(\mathbf{W}_{\sf RF}^{(n)} ,\mathbf{W}_{\sf BB}^{(n)}, \mathbf{\Theta}^{(n)}\right), t=1,2,\dots\right\}$, and the covert rate has an upper bound due to the limited transmit power. Hence, the Algorithm~\ref{alg4} is guaranteed to converge. 

Notice that the AO-based approach only converges to a suboptimal solution in general. The reason is as follows. The problem \eqref{eq:overall_pro} is not jointly convex with respect to $\mathbf{W}_{\sf RF}$, $\mathbf{W}_{\sf BB}$ and $\mathbf{\Theta}$, and the WMMSE method together with the ADMM algorithm for solving the problem \eqref{eq:overall_pro} returns only suboptimal solution.

Denote the total numbers of iterations for Algorithms~\ref{alg1_2}, \ref{alg2}, \ref{alg3}, \ref{alg4} are $t_1$, $t_2$, $t_3$ and $T$, respectively. The computation complexity for Algorithm~\ref{alg1_2} that relies on the inverse matrix operation, is on the order of $\mathcal{O}\left( M_{\sf A}^3\right)$~\cite{pan2020intelligent}, and the number of iterations for bisection search to get converged is $\log_2\left( \frac{\mu_{\sf u} - \mu_{\sf l}}{\varepsilon}\right)\log_2\left( \frac{\upsilon_{\sf u} - \upsilon_{\sf l}}{\varepsilon}\right)$~\cite{pan2020intelligent}. Thus, the whole computational complexity for Algorithm~\ref{alg1_2} is given by $\mathcal{O}\left( t_1\log_2\left( \frac{\mu_{\sf u} - \mu_{\sf l}}{\varepsilon}\right)\log_2\left( \frac{\upsilon_{\sf u} - \upsilon_{\sf l}}{\varepsilon}\right)M_{\sf A}^3\right)$. The computational complexity for Algorithm~\ref{alg2} mainly relies on the MO algorithm, and is given by $\mathcal{O}\left(M_{\sf A}^2 \right)$~\cite{guo2020weighted}. Then, the whole computational complexity for Algorithm~\ref{alg2} is given by $\mathcal{O}\left( t_2M_{\sf A}^2\right)$. Similarly, the complexity for Algorithm~\ref{alg3} mainly relies on solving sub-problem~\eqref{eq:update_V}, which is given by $\mathcal{O}\left(N^{3.5} \right)$. The whole computational complexity for Algorithm~\ref{alg3} is given by $\mathcal{O}\left( t_3N^{3.5}\right)$~\cite{hong2016convergence}, which is much lower than the computational complexity $\mathcal{O}\left( t_3N^{7}\right)$ arising from the SDR-based algorithm~\cite{wei2023secure,ge2022reconfigurable}.  Hence, the total computational complexity for Algorithm~\ref{alg4} is $\mathcal{O}\left(T\left( t_1\log_2\left( \frac{\mu_{\sf u} - \mu_{\sf l}}{\varepsilon}\right)\log_2\left( \frac{\upsilon_{\sf u} - \upsilon_{\sf l}}{\varepsilon}\right)M_{\sf A}^3 + t_2M_{\sf A}^2 + t_3N^{3.5}   \right)\right)$.

\section{Beam Diffraction in XL-RIS-Assisted System}\label{sec:Beam_diff}

This section reveals a unique beam pattern, referred to as beam diffraction, in the XL-RIS-assisted near-field system. A mathematical analysis is conducted to explain the existence of this beam pattern.


In far-field communications, beamforming is a multi-antenna technique that steers the wireless signal toward the intended receiver by manipulating antenna elements. This produces a radiation pattern similar to a flashlight beam at a particular angle. In other words, beamforming exhibits DoF solely in the angular domain, because the array response vector in the far-field is only determined by the angle.


In near-field communications, a new beamforming paradigm known as beam focusing was introduced in~\cite{liu2023near,zhang2022beam,zhang20236g}. The radiation pattern can focus on a specific location in the polar domain. This effect is caused by the asymptotic orthogonality exhibited by the near-field array response vector in both distance and angular domains with an increasing number of antennas $M$~\cite[Appendix A]{wu2023multiple}, which can be expressed as follows 
\begin{equation} \label{eq:asy_ort}
\begin{aligned}
    & \mathop {\lim }\limits_{M \to  + \infty } \left| {{{\mathbf{a}}^{\sf{H}}}\left( {{\phi _1},{d_1}} \right) {\mathbf{a}}\left( {{\phi _2},{d_2}} \right)} \right| = 0, \\ & \quad \quad \quad {\text{ for }}{\phi _1} \ne {\phi _2}{\text{ or }}{d_1} \ne {d_2},
\end{aligned}
\end{equation}
where $\phi_1$, $\phi_2$, $d_1$, $d_2$ denote the angle and distance in the polar domain, respectively. And this property implies that beam focusing has DoF in both angular and distance domains.

Clearly, for a traditional wireless system where Bob is within Alice's near-field region, the benefits of asymptotic orthogonality can be exploited to achieve a positive covert rate by employing the maximum ratio transmission (MRT) at Alice with power control. However, for XL-RIS-assisted communication systems where Bob is located within the near-field range of XL-RIS, we have the following theorem.
\begin{theorem}\label{theorem1}
     The equivalent channels in near-field XL-RIS-assisted communication systems do not satisfy asymptotic orthogonality, and the channel vectors are also not collinear.
\end{theorem} 
\begin{proof}
    For ease of description, we consider a basic scenario where the receivers are equipped with a single antenna\footnote{Our proof can be extended to the scenario that the receivers are equipped with multiple antennas, and the simulation results in Fig.~\ref{fig:beam_pattern} are obtained for the scenario of multiple antennas.}, and there are only LoS paths between the XL-RIS and the receivers. Thus, the channels for two receivers in distinct locations can be given, respectively,
\begin{equation}
\begin{aligned}
      &  \mathbf{h}_1 = \rho_1\mathbf{a}_1\left( {{\phi _1},{d_1}} \right), \\
   & \mathbf{h}_2 = \rho_2\mathbf{a}_2\left( {{\phi _2},{d_2}} \right), \\
    &   {\text{ for }}{\phi _1} \ne {\phi _2}{\text{ or }}{d_1} \ne {d_2},
\end{aligned}
\end{equation}
where $\rho_1$ and $\rho_2$ denote the path loss of the two receivers, respectively, and $\mathbf{a}_1\left( {{\phi _1},{d_1}} \right)$ and $\mathbf{a}_2\left( {{\phi_2},{d_2}} \right)$ are the corresponding array response vectors, which satisfy the asymptotic orthogonality in~\eqref{eq:asy_ort}. Then, the equivalent channels for the two receivers are given as 
\begin{equation}
\begin{aligned}
        \mathbf{\bar h}_1^{\sf{H}} & = \mathbf{h}_1^{\sf{H}}\mathbf{\Theta}\mathbf{G}, \\
         \mathbf{\bar h}_2^{\sf{H}} & = \mathbf{h}_2^{\sf{H}}\mathbf{\Theta}\mathbf{G}.
\end{aligned}
\end{equation}
The inner product of the two channels is expressed as
\begin{equation}
    \left| \mathbf{\bar h}_1^{\sf{H}} \mathbf{\bar h}_2 \right| = {\rho _1}{\rho _2}\left| {{\mathbf{a}}_1^{\sf{H}}\left( {{\phi _1},{d_1}} \right){\mathbf{\Theta G}}{{\mathbf{G}}^{\sf{H}}}{{\mathbf{\Theta }}^{\sf{H}}}{{\mathbf{a}}_2}\left( {{\phi _2},{d_2}} \right)} \right|.
\end{equation}
If the two channels satisfy the asymptotic orthogonality, i.e., $\mathop {\lim }\limits_{M \to  + \infty } \left| \mathbf{\bar h}_1^{\sf{H}} \mathbf{\bar h}_2 \right| = 0$, then ${\mathbf{\Theta G}}{{\mathbf{G}}^{\sf{H}}}{{\mathbf{\Theta}}^{\sf{H}}}$ should be a scaled identity matrix $\nu\mathbf{I}_M$, where $\nu$ is the scaled factor.

Since $\mathbf{G}\mathbf{G}^{\sf H}\in \mathbb{C}^{N \times N}$ is a Hermitian matrix,  it can be decomposed by the eigenvalue decomposition as follows 
\begin{equation}
    \mathbf{G}\mathbf{G}^{{\sf{H}}} = \mathbf{Q}\mathbf{\Lambda}\mathbf{Q}^{{\sf{H}}} =\sum\limits_{i=1}^N {{\lambda_i}{\mathbf{q}_i}{{\mathbf{q}_i^{\sf{H}}}}},
\end{equation}
where  $\mathbf{Q} = \left[\mathbf{q}_1,\dots,\mathbf{q}_N \right]\in\mathbb{C}^{N\times N}$, and $\lambda_i$ and $\mathbf{q}_i, \forall i=1,\dots N$, are the eigenvalues and the corresponding eigenvectors of the matrix. Hence, we have the following equation
\begin{equation} \label{eq:TGGT}
\begin{aligned}
 &   {\mathbf{\Theta G}}{{\mathbf{G}}^{\sf{H}}}{{\mathbf{\Theta }}^{\sf{H}}} =  \\ 
 & \left[ {\begin{array}{*{20}{c}}
  {\sum\nolimits_{i = 1}^N {{\lambda _i}{{\left| {{q_{1,i}}} \right|}^2}} }& \ldots &{\sum\nolimits_{i = 1}^N {{\lambda _i}{e^{\jmath {\theta_1-\theta_N}}}{q_{1,i}}q_{N,i}^*} } \\ 
   \vdots & \ddots & \vdots  \\ 
  {\sum\nolimits_{i = 1}^N {{\lambda _i}{e^{\jmath {\theta _N-\theta_1}}}{q_{N,i}}q_{1,i}^*} }& \ldots &{\sum\nolimits_{i = 1}^N {{\lambda _i}{{\left| {{q_{N,i}}} \right|}^2}} } 
\end{array}} \right], 
\end{aligned}    
\end{equation}
where $q_{i,j}$ is the $i$-th row and $j$-th column element of $\mathbf{Q}$. From~\eqref{eq:TGGT}, we find that the diagonal elements are not equal and the off-diagonal elements are also not zero unless $\mathbf{G}\mathbf{G}^{\sf H}$ itself is a scaled identity matrix. The matrix ${\mathbf{\Theta G}}{{\mathbf{G}}^{\sf{H}}}{{\mathbf{\Theta }}^{\sf{H}}}$ can not be set as a scaled identity matrix by only adjusting the coefficients of the XL-RIS.

Two equivalent channel vectors are colinear, if they satisfy the condition $\mathbf{\bar h}_1^{\sf{H}} = \alpha \mathbf{\bar h}_2^{\sf{H}}$, where $\alpha$ is a scaling factor. Consequently, we need $\left(\mathbf{h}_1^{\sf{H}} - \alpha\mathbf{h}_2^{\sf{H}}\right)\mathbf{\Theta}\mathbf{G}=0$. However, the two channels $\mathbf{h}_1$ and $\mathbf{h}_2$ satisfy the  asymptotic orthogonality, hence $\left(\mathbf{h}_1^{\sf{H}} - \alpha\mathbf{h}_2^{\sf{H}}\right)$ cannot be equal to $\mathbf{0}$. In other words, $\mathbf{\bar h}_1^{\sf{H}}$ cannot be equal to $\alpha \mathbf{\bar h}_2^{\sf{H}}$ for arbitrary $\alpha$. Hence, the proof is complete.
\end{proof}


 

Based on the above theorem, the equivalent near-field channels in the XL-RIS-assisted communication system exhibit linear independence rather than asymptotic orthogonality, which is different from the properties of traditional far-field and near-field channels. Consequently, a novel beam pattern termed beam diffraction emerges, representing a transitional state between beamforming and beam focusing. Specifically, this beam pattern allows for bypassing non-target receivers at the same angle before converging on the intended user or vice versa.

Particularly, the beam can be focused on a specific area in XL-RIS-assisted systems if a certain requirement is satisfied, which is given in Corollary~\ref{corollary2}.
\begin{corollary} \label{corollary2}
     Let $\mathbf{a} = \left[a_1,\dots,a_N \right]^{\sf T} \triangleq \mathbf{Gw}$ and $ = \mathbf{b} = \left[b_1,\dots,b_{N} \right]^{\sf T} \triangleq {\rm{diag}}(\mathbf{h}_1^{\sf H})\mathbf{Gw}$, where $\mathbf{h}_1$ represents the channel between the XL-RIS and Bob. Then, based on the Theorem~\ref{theorem1}, the beam can focus on a specific area if the following requirement is satisfied
   \begin{equation}
       \frac{{{{\left| {{a_1}} \right|}^2}}}{{\left| {{b_1}} \right|}} = \frac{{{{\left| {{a_2}} \right|}^2}}}{{\left| {{b_2}} \right|}} = ... = \frac{{{{\left| {{a_N}} \right|}^2}}}{{\left| {{b_N}} \right|}}.
   \end{equation}
\end{corollary}
\begin{proof}
    Please refer to Appendix~\ref{App2}.
\end{proof}

It is worth noting that beam focusing is not a good strategy in XL-RIS-assisted covert systems because the covert rate will degrade when the requirement in Corollary~\ref{corollary2} is satisfied. The reason is as follows. We need to jointly optimize $\mathbf{w}$ and $\mathbf{v}$ so that $\mathbf{w}$ is in the orthogonal space of Willie's equivalent channel. However, this required precoding strategy to achieve beam focusing is too restrictive.

\section{Numerical Results} \label{Sub:Nem_res}

This section presents numerical results to evaluate the performance of the proposed algorithm in an XL-RIS empowered covert communication system. The system deployment is depicted in Fig.~\ref{fig:sim_env}. Considering a three-dimensional Cartesian coordinate system, the XL-RIS is deployed on the $yz$-plane with its center located at the origin of coordinate $(0,0,0)$, Alice, Bob, and Willie are positioned on the $xy$-plane. For convenience, the locations of Alice, Bob, and Willie can be represented using a polar coordinate system. Specifically, Alice is located at $(50\text{ m}, -\pi/4)$. In this simulation, we consider an extreme scenario where Willie is located between the XL-RIS and the Bob with identical azimuth angles. Unless otherwise specified, the locations of Bob and Willie are set as  $(15\text{ m}, \pi/4)$ and  $(10\text{ m}, \pi/4)$, respectively. We assume that the links between Alice and XL-RIS are LoS dominated and the corresponding path loss is modeled as $69.4 + 24.0\log_{10}(D_{\sf{AX}})$ dB~\cite{wang2021covert}, where $D_{\sf{AX}}$ denotes the distance in meter. The other default parameters are set as $M_{\sf A} = 64$, $M_{\sf B} = 4$, $M_{\sf W} = 4$, $N_{\sf y}=90$, $N_{\sf z}=8$~\cite{yu2023channel,shen2023multi,miridakis2022zero}, $M_{\sf A}^{\sf RF}=4$, $L = 2$, $N_{\sf c}=5$, $N_{\sf ray}=10$, $f = 28 \text{ GHz} $~\cite{cheng2023achievable}, $\sigma_{\sf B}^2 = -110 \text{ dBm}$, $\sigma_{\sf W}^2 = -110 \text{ dBm}$, $d_{\sf A}=\frac{\lambda}{2}$, $d_{\sf X}=\frac{\lambda}{2}$, $\rho=3 \text{ dB}$, $\kappa=0.01$, $P_{\max}=40\text{ dBm}$. As the Rayleigh distance of the XL-RIS is computed as $\frac{2D^2}{\lambda}=42.78$ m, the far-filed channel model is adopted for the Alice-to-XL-RIS channel. It is noted that our proposed algorithm is also applicable when Alice and XL-RIS are in the near field of each other, as the channel coefficients are treated as known parameters in the optimization problem. The numerical results are obtained by averaging over 100 random channel realizations.

\begin{figure}[t]
    \centering
    \includegraphics[width=0.70\linewidth]{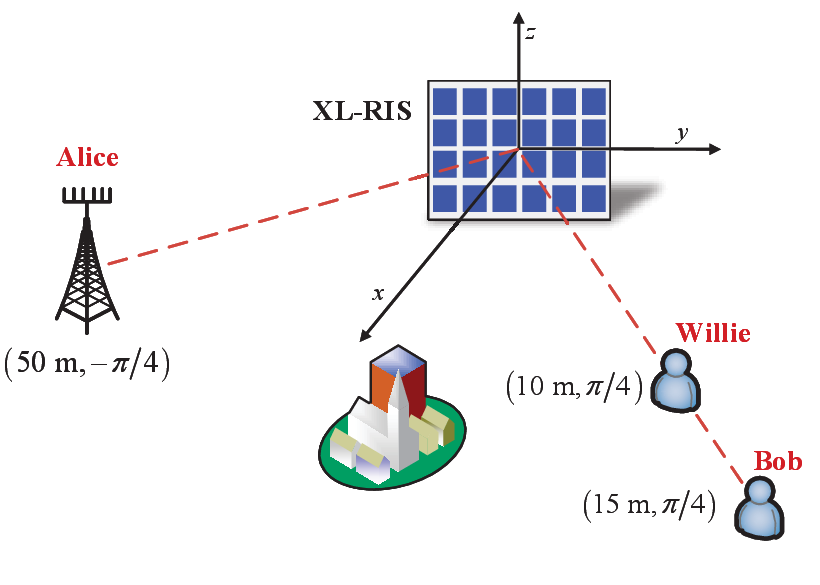}
    \caption{Deployment of an XL-RIS empowered near-field covert communication system.}
    \label{fig:sim_env}
\end{figure}

 In this paper, we compare our proposed algorithm for hybrid beamforming architectures (abbr., Proposed HB) with the following benchmark schemes: 1) Fully-digital (abbr., FD): Alice adopts a fully-digital beamformer $\mathbf{W}_{\sf FD}$ , which serves as an upper bound for the proposed algorithm; 2) Random phase (abbr., RP): The phase of XL-RIS's reflection coefficient is randomly generated from the interval $(0,2\pi]$; 3) Far-field (abbr., FF): The channel between the XL-RIS and the receivers (Bob and Willie) are modeled as conventional far-field channel model; 4) Zero-forcing (abbr., ZF): The ZF beamform is utilized to null the received power of useful signal at Willie, and the steps are as follows. For fixed $\mathbf{\Theta}$, we take the first $L$ columns of the right singular value matrix of $\mathbf{H}_{\sf B}$ as the optimal fully-digital beamformer at Alice, which is denoted by $\mathbf{\widetilde W}_{\sf FD}$. Then, we project $\mathbf{\widetilde W}_{\sf FD}$ onto the orthogonal complement of $\mathbf{H}_{\sf W}$, which is given as
 \begin{equation}
     {{\mathbf{W}}_{{\sf{FD}}}^{\sf ZF}} = \sqrt {{P_{\max }}} \frac{{\left( {{{\mathbf{I}}_{M_{\sf A}}} - {\mathbf{H}}_{\sf{W}}^\dag {{\mathbf{H}}_{\sf{W}}}} \right){{\widetilde {\mathbf{W}}}_{{\sf{FD}}}}}}{{{{\left\| {\left( {{\mathbf{I}_{M_{\sf A}}} - {\mathbf{H}}_{\sf{W}}^\dag {{\mathbf{H}}_{\sf{W}}}} \right){{\widetilde {\mathbf{W}}}_{{\sf{FD}}}}} \right\|}_{\sf{F}}}}},
 \end{equation}
 where $\left( \cdot\right)^\dag$ denotes the Moore–Penrose inverse. Once ${{\mathbf{W}}_{{\sf{FD}}}^{\sf ZF}}$ is obtained, we optimize $\mathbf{\Theta}$ according to the sum-path-gain maximization criterion~\cite{ning2020beamforming} as follows
 \begin{maxi!}
 {\mathbf{\Theta }}{{\text{ Tr}}\left( {{\mathbf{H\Theta G}}{{\mathbf{W}}_{{\sf{RF}}}}{{\mathbf{W}}_{{\sf{BB}}}}{\mathbf{W}}_{\sf{BB}}^{\sf H}{\mathbf{W}}_{\sf{RF}}^{\sf H}{{\mathbf{G}}^{\sf H}}{{\mathbf{\Theta }}^{\sf H}}{{\mathbf{H}}^{\sf H}}} \right)}{}{}
 \addConstraint{\eqref{eq:OP_RISC},}{}
 \end{maxi!}
 which can be solved by projected gradient ascent~\cite{kammoun2020asymptotic}.

Fig.~\ref{fig:conv} plots the convergence behaviors of the proposed algorithm for different Willie locations. It can be observed that the covert rate at Bob increases monotonically and saturates within a few number of iterations. On the one hand, interestingly, for fixed Bob location $(15 \text{ m}, \pi/4)$, when Willie is closer to the XL-RIS and has the same angle as Bob's, e.g., the Willie is located at $(10 \text{ m}, \pi/4)$, a positive covert rate can still be achieved, which is not realizable for far-field covert communication systems. As Willie moves farther away from the XL-RIS while maintaining the same azimuth angle as Bob's, the covert rate at Bob decreases. On the other hand, a higher covert rate is achieved, when Willie is located at $(15 \text{ m}, \pi/6)$ and has different azimuth angle from Bob. This implies the covert rate improvement from angle difference between Bob and Willie, which is similar to the results for far-field communications. 


\begin{figure}[t]
    \centering
    \includegraphics[width=0.8\linewidth]{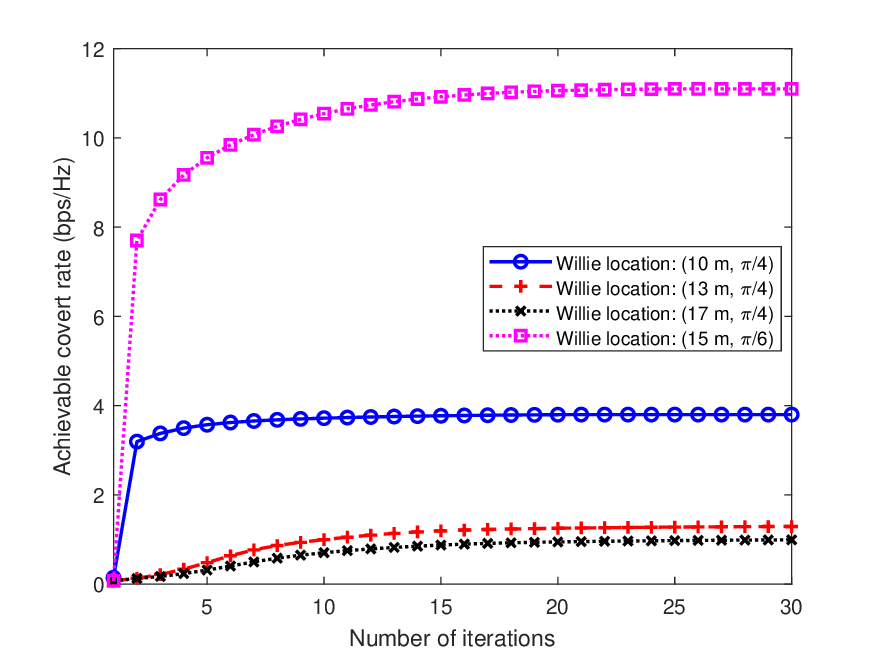}
    \caption{Convergence behaviors of the proposed algorithm.}
    \label{fig:conv}
\end{figure}

\begin{figure*}[t]
	\centering  
	\subfigbottomskip=2pt 
	\subfigcapskip=-5pt 
	\subfigure[Normalized heat map for beam steering in legacy far-field communication systems.]{
		\includegraphics[width=0.30\linewidth]{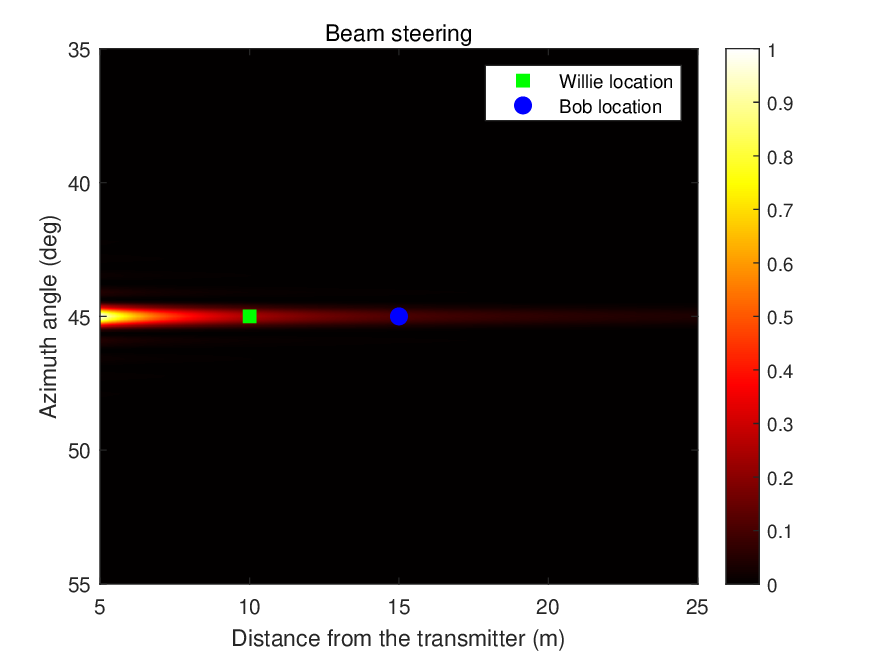} \label{fig:beamforming}}
	\subfigure[Normalized heat map for beam focusing in XL-RIS empowered near-field communication systems.]{
		\includegraphics[width=0.30\linewidth]{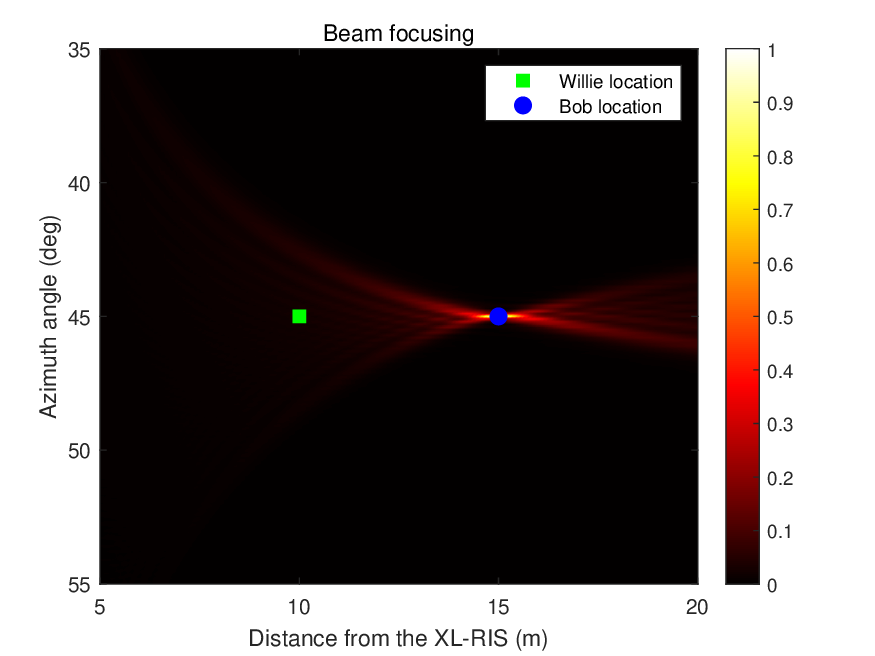} \label{fig:beam_focusing}}
	\subfigure[Normalized heat map for beam diffraction in XL-RIS empowered near-field communication systems, when Willie is located at $(10\text{ m}, 45^\circ)$.]{
		\includegraphics[width=0.30\linewidth]{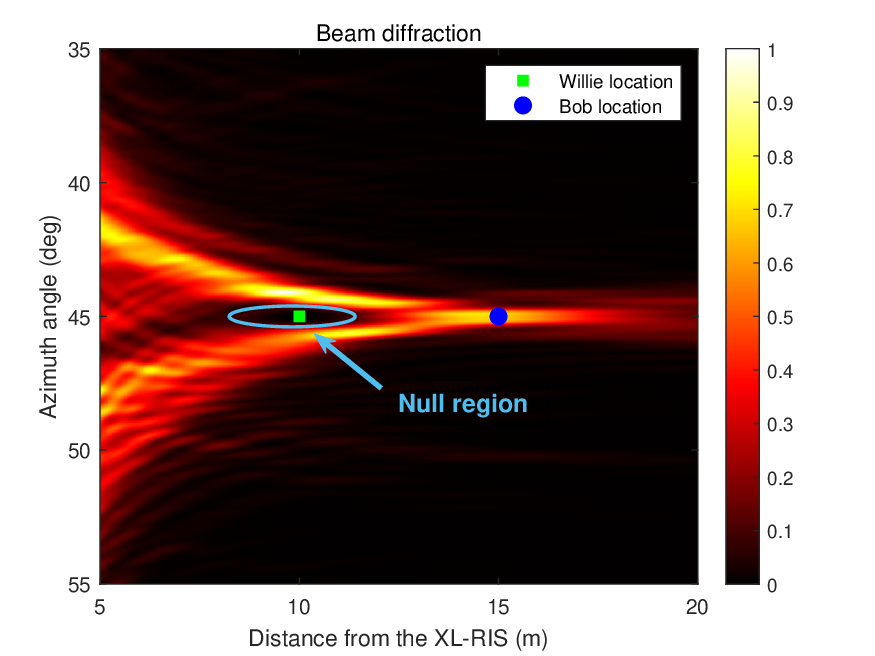} \label{fig:beam_diffraction1}}
	\subfigure[Normalized heat map for beam diffraction in XL-RIS empowered near-field communication systems, when Willie is located at $(20\text{ m}, 45^\circ)$.]{
		\includegraphics[width=0.30\linewidth]{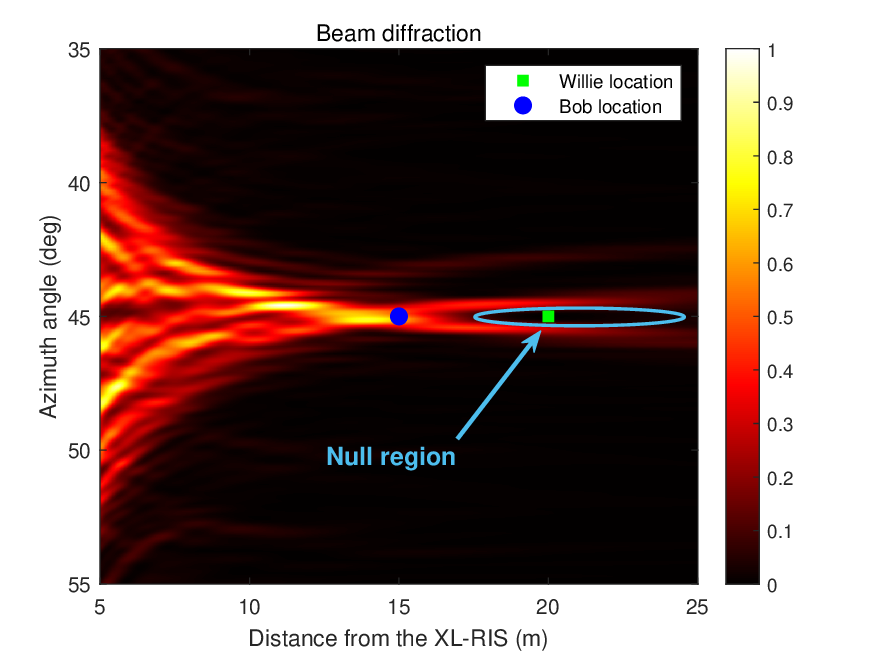} \label{fig:beam_diffraction2}}
        \subfigure[Normalized heat map for beam diffraction in XL-RIS empowered near-field communication systems, when Willie is located at $(8\text{ m}, 43^\circ)$.]{
		\includegraphics[width=0.30\linewidth]{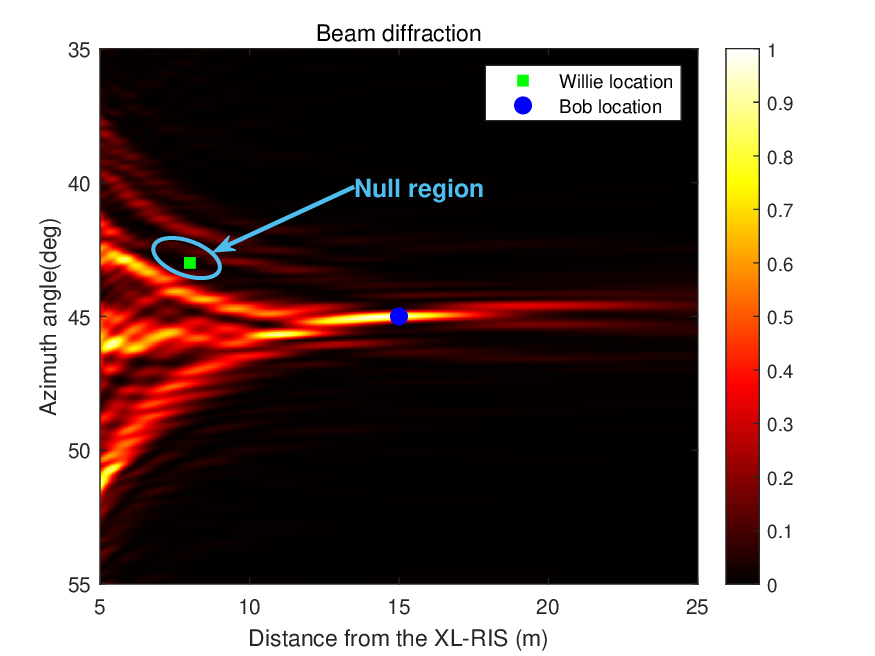} \label{fig:beam_diffraction3}}
          \subfigure[Normalized heat map for beam diffraction in XL-RIS empowered near-field communication systems, when Willie is located at $(8\text{ m}, 47^\circ)$.]{
		\includegraphics[width=0.30\linewidth]{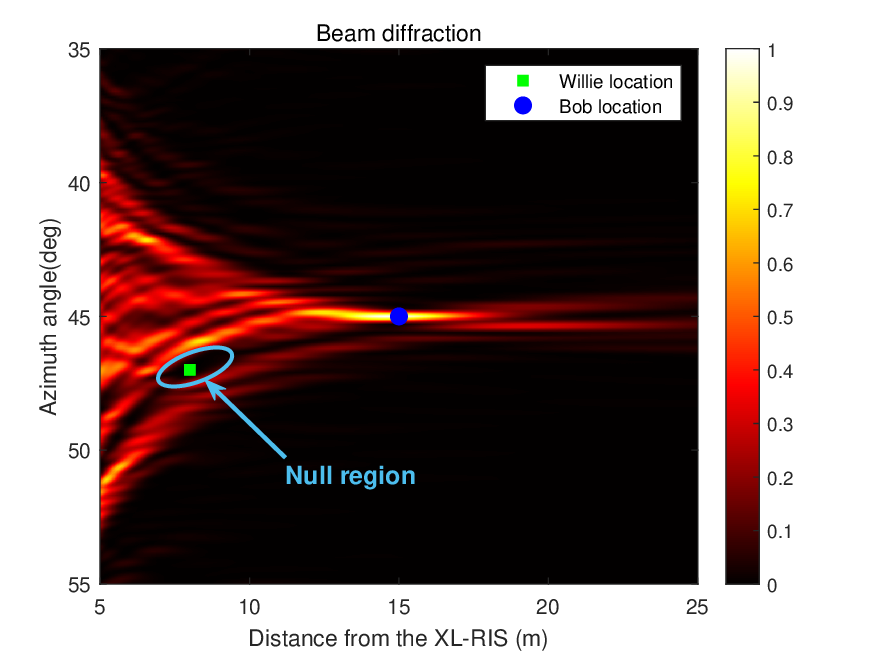} \label{fig:beam_diffraction4}}
  	\caption{Normalized heat maps for beam steering, beam focusing, and beam diffraction.} 
        \label{fig:beam_pattern}
\end{figure*}

Fig.~\ref{fig:beam_pattern} illustrates the normalized power heat-maps depicting three beam patterns. First, Fig.~\ref{fig:beamforming} shows the beam steering pattern for far-field communication systems. With the preseted user locations, Willie always has higher received power than Bob, thus it is impossible for Bob to achieve a positive covert rate. In a general case, the covert rate is also severely limited when Willie has better channel quality than Bob. Second, Fig.~\ref{fig:beam_focusing} shows the beam focusing pattern in the XL-RIS empowered near-field systems. The beam focusing is optimal for near-filed communication systems without RIS assistance, due to the channels' asymptotic orthogonality. For the considered XL-RIS empowered near-filed communication system, the equivalent cascade channels are not asymptotically orthogonal in general. The beam pattern in Fig.~\ref{fig:beam_focusing} is obtained by forcing the XL-RIS reflection coefficients to ensure asymptotic orthogonality of the equivalent cascade channels. Third, using the proposed algorithm, the beam pattern for different Willie locations are illustrated in Fig.~\ref{fig:beam_diffraction1}, Fig.~\ref{fig:beam_diffraction2}, Fig.~\ref{fig:beam_diffraction3}, and Fig.~\ref{fig:beam_diffraction4}. From numerical simulations, the covert rate achieved by the proposed beam pattern in Fig.~\ref{fig:beam_diffraction1} is 3.799 bps/Hz, which is significantly higher than the covert rate of 0.342 bps/Hz achieved by the beam focusing pattern in Fig.~\ref{fig:beam_focusing}. This shows that the proposed beam pattern is promising for XL-RIS empowered near-filed covert communications, while the beam focusing pattern suffers from poor covert transmission performance. 

In details, observed from Fig.~\ref{fig:beam_diffraction1}, Fig.~\ref{fig:beam_diffraction2}, Fig.~\ref{fig:beam_diffraction3} and Fig.~\ref{fig:beam_diffraction4}, regardless of whether Willie is closer to or farther from XL-RIS than Bob, or at the same azimuth angle as Bob, a \emph{null region} is created around Willie's location in the beam pattern. In other words, the proposed beam pattern can split and bypass the non-target user Willie, and converge on the intended user Bob, which leads to significant covert-rate enhancement. We refer to this beam pattern observed in the XL-RIS empowered covert communication systems as ``\emph{beam diffraction}'', which represents a transitional state between beam steering and beam focusing.





\begin{figure}[t]
    \centering
    \includegraphics[width=0.90\linewidth]{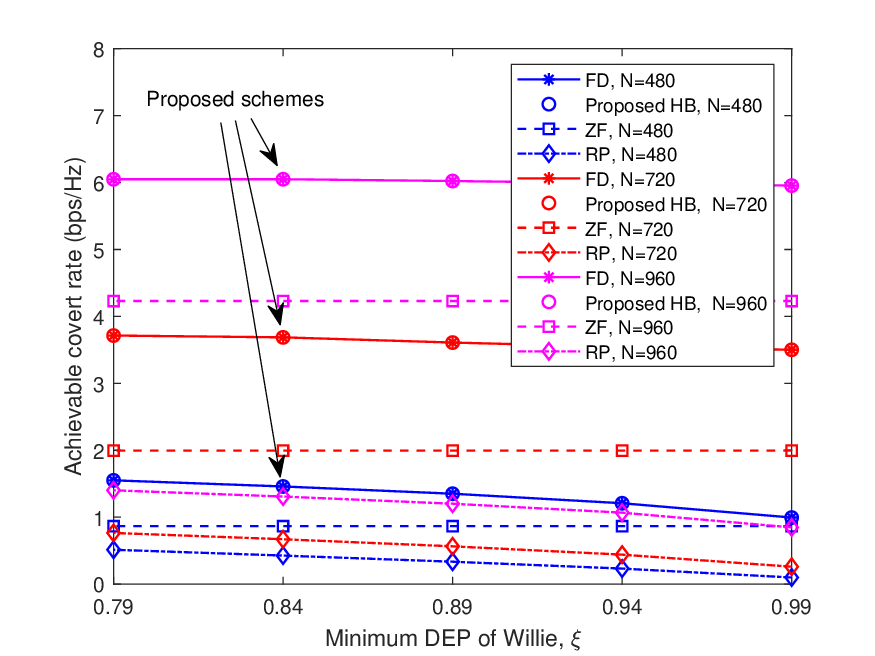}
    \caption{Covert rate versus minimum DEP requirement of Willie, $\xi$.}
    \label{fig:xi}
\end{figure}

Fig.~\ref{fig:xi} plots the covert rate performance against the minimum DEP requirement $\xi$. First, it is observed that the achievable covert rate of all schemes decreases as $\xi$ increases. Second, the performance achieved by the optimized hybrid beamformers is comparable to that achieved by the fully digital beamformers. Third, the proposed algorithm achieves superior performance compared with ZF and RP schemes, for fixed number of reflection elements $N$. Additionally, the covert-rate performance degradation becomes negligible as $\xi$ increases, for larger $N$. For instance, for $N=960$, increasing $\xi$ from 0.79 to 0.99 only results into a decrease of 0.094 bps/Hz, whereas for $N=480$, it results into a larger decrease of 0.556 bps/Hz. This observation is due to the beam diffraction pattern, which is capable of creating a null region around Willie and thus results in minimal power leakage to Willie. Therefore, reducing Willie's DEP $\xi$ will not significantly affect the covert rate. Additionally, as the number of XL-RIS elements grows, the diffraction capability improves. Consequently, the covert rate becomes even less affected by Willie's DEP.

\begin{figure}[t]
    \centering
    \includegraphics[width=0.90\linewidth]{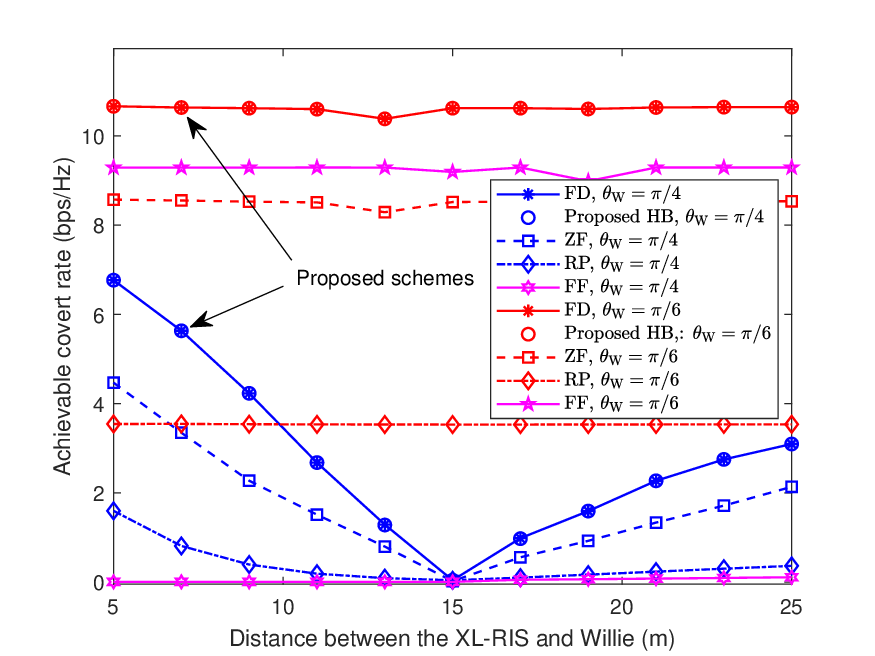}
    \caption{Covert rate versus distance between the XL-RIS and Willie.}
    \label{fig:dis}
\end{figure}


Fig.~\ref{fig:dis} plots the covert rate performance against the distance between the XL-RIS and Willie for different azimuth angles. On the one hand, when Willie is at the same azimuth angle $\pi/4$ as Bob, the proposed algorithm can always achieve satisfactory covert rate performance in near-field communication systems; in contrast, the covert rate is almost zero in far-field communication systems. In details, the covert rate decreases to zero as the location of Willie moves from (5m, $\pi/4$) to the Bob location (15m, $\pi/4$), and then increases as the location of Willie further moves far away from Bob. This indicates that the closer to the XL-RIS location (0, 0), the stronger beam diffraction ability of the XL-RIS. On the other hand, When Willie is at the azimuth angle $\pi/6$, being different from Bob's location angle $\pi/4$, the distance between the XL-RIS and Willie has negligible impact on Bob's covert rate obtained by all schemes, including far-field communications. 
\begin{figure}[t]
    \centering
    \includegraphics[width=0.9\linewidth]{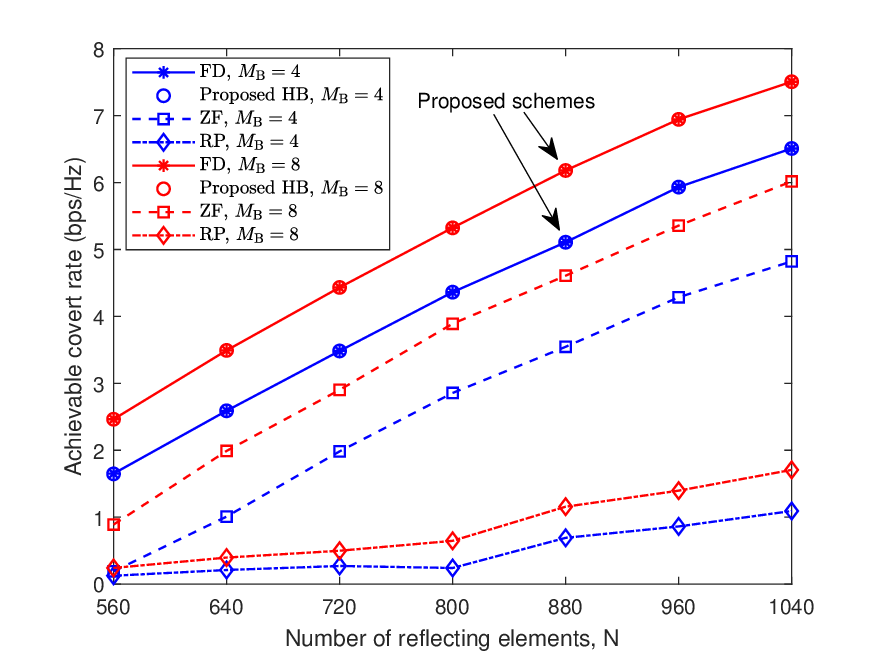}
    \caption{Covert rate versus number of XL-RIS's elements, $N$.}
    \label{fig:N}
\end{figure}

Fig.~\ref{fig:N} plots the covert rate performance against the number of reflection elements $N$ of the XL-RIS. The covert rate achieved by all algorithms increases as $N$ grows, indicating that a larger antenna aperture can enhance the covert communication performance. The proposed algorithm achieves significant performance improvement compared to the benchmark ZF and RP schemes. For example, when $N=880$ and $M_{\sf B}=4$, compared to RP, the covert rate performance gain reaches up to 5.071 bps/Hz, which emphasizes the importance of optimizing the XL-RIS's reflection coefficients. Also, equipping Bob with more antennas leads to increase in the covert rate.

\begin{figure}[t]
    \centering
    \includegraphics[width=0.9\linewidth]{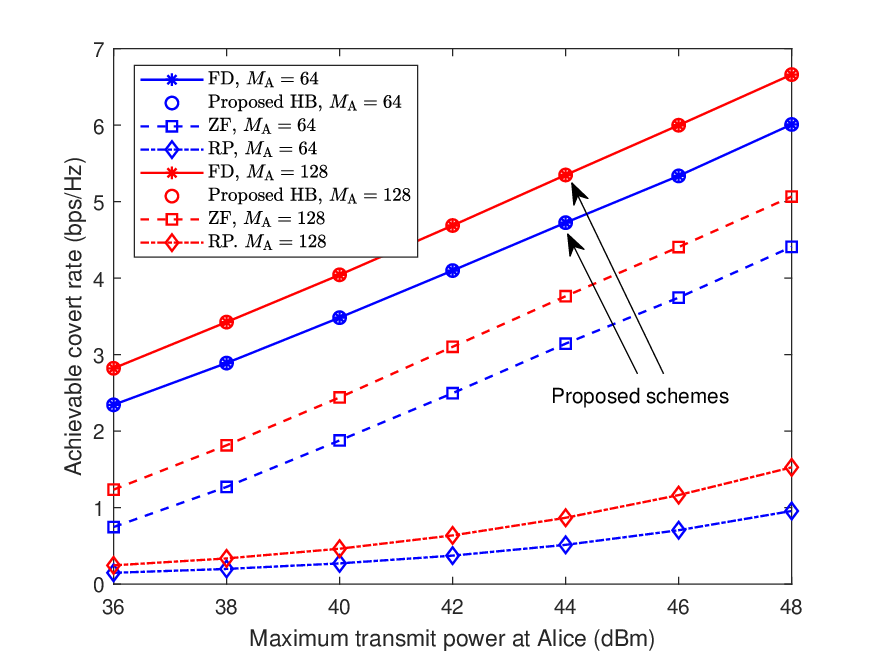}
    \caption{Covert rate versus maximum transmit power at Alice $P_{\max}$.}
    \label{fig:Pmax}
\end{figure}

Fig.~\ref{fig:Pmax} plots the covert rate performance against the maximum transmit power $P_{\max}$ at Alice. The covert rate monotonically increases as $P_{\max}$ increases for all algorithms. The proposed algorithm achieves significant covert rate performance improvement than the benchmarks. For instance, for $P_{\max}=40~ \text{dBm}$ and $M_A=64$, the covert rate is enhanced by 113\%, compared to the ZF benchmark. Also, increasing Alice's number of antennas $M_A$ can enhance the covert rate. 



\section{Conclusion}\label{Sub:Conc}

In this paper, we studied an XL-RIS empowered near-field covert communication system, and maximized the Bob's achievable covert rate under the Willie's covertness constraint. An AO-based algorithm was designed to alternatively optimize the hybrid beamformers at Alice, as well as the reflection coefficient matrix at the XL-RIS. Numerical results verified the effectiveness of the optimized design. By exploiting both the distance and angular DoFs, the near-field communications achieve higher covert rates than the far-field covert communications. More interesting, even for the extreme scenario where Willie is located at the same direction as Bob and closer to the XL-RIS, the covert transmission can still be achieved for the near-field communications, which is not realizable for far-filed communications. Moreover, with the proposed optimal design, the beam in the XL-RIS empowered near-field system can split and bypass Willie for covert transmission to Bob, referred to \emph{beam diffraction} phenomenon. This beam pattern can be extended to other XL-RIS-assisted near-field systems, such as multi-user or physical layer security systems. Although this paper assumes perfect Willie-involved CSI available at Alice, the robust covert communication design under imperfectly-or-partially-known CSI can be further studied for the XL-RIS empowered near-field covert communication system, which is out of the scope of this paper due to space limitation. 

\appendices

\section{Proof of Corollary~\ref{corollary2}} \label{App2}
Assume that $\mathbf{h}_1$ and $\mathbf{h}_2$ are the channels between the XL-RIS and Bob and the XL-RIS and Willie, respectively, and there are only LoS paths between the XL-RIS and the receivers. The received signals at Bob and Willie are given as
\begin{equation} \label{eq:coro1}
\begin{aligned}
        {y_1} & = {\mathbf{h}}_1^{\sf H}{\mathbf{\Theta Gw}} + {n_1} = {{\mathbf{v}}^{\sf H}}\text{diag}({\mathbf{h}}_1^{\sf H}){\mathbf{Gw}} + {n_1},\\
      {y_2} & = {\mathbf{h}}_2^{\sf H}{\mathbf{\Theta Gw}} + {n_2} = {{\mathbf{v}}^{\sf H}}\text{diag}({\mathbf{h}}_2^{\sf H}){\mathbf{Gw}} + {n_2}. 
\end{aligned}
\end{equation}
According to the phase alignment algorithm and the unit-modulus constraints, the optimal reflection coefficient vector for Bob at the XL-RIS is given as
\begin{equation}
    v_n = \frac{\left(\text{diag}({\mathbf{h}}_1^{\sf H}){\mathbf{Gw}} \right)_n}{\left| \text{diag}({\mathbf{h}}_1^{\sf H}){\mathbf{Gw}} \right|_n},
\end{equation}
where $v_n$ is the $n$-th element of $\mathbf{v}$. Then, substituting $\mathbf{v}$ into \eqref{eq:coro1}, yields
\begin{equation}
\begin{aligned}
y_2  
   & = {\left[ {\frac{{{{\left( {{\text{diag}}({\mathbf{h}}_1^{\sf{H}}){\mathbf{Gw}}} \right)}_1}}}{{{{\left| {{\text{diag}}({\mathbf{h}}_1^{\text{H}}){\mathbf{Gw}}} \right|}_1}}}, \ldots ,\frac{{{{\left( {{\text{diag}}({\mathbf{h}}_1^{\sf{H}}){\mathbf{Gw}}} \right)}_N}}}{{{{\left| {{\text{diag}}({\mathbf{h}}_1^{\sf{H}}){\mathbf{Gw}}} \right|}_N}}}} \right]^{\text{*}}}{\text{diag}}({\mathbf{h}}_2^{\sf{H}}){\mathbf{Gw}}  \\ & \quad + {n_2} \hfill \\
  &  = \frac{{{h_{11}}h_{21}^*{{\left| {{a_1}} \right|}^2}}}{{\left| {{b_1}} \right|}} + \frac{{{h_{12}}h_{22}^*{{\left| {{a_2}} \right|}^2}}}{{\left| {{b_2}} \right|}} \ldots  + \frac{{{h_{1N}}h_{2N}^*{{\left| {{a_N}} \right|}^2}}}{{\left| {{b_N}} \right|}} + {n_2}, \hfill \\ 
\end{aligned}
\end{equation}
where $h_{1i}$ and $h_{2i}$ denote the $i$-th elements of $\mathbf{h}_1$ and $\mathbf{h}_2$, respectively. Since $\mathop {\lim }\limits_{N \to  + \infty } \sum\nolimits_{i = 1}^N {{h_{1i}}h_{2i}^*}  = 0$, when the coefficients $\frac{\left| a_i\right|^2}{\left| b_i\right|}, i=1,\ldots,N$ are equal, we obtain 
\begin{equation} \label{eq:beam_focus_req}
   \mathop {\lim }\limits_{N \to  + \infty } \frac{{{h_{11}}h_{21}^*{{\left| {{a_1}} \right|}^2}}}{{\left| {{b_1}} \right|}} + \frac{{{h_{12}}h_{22}^*{{\left| {{a_2}} \right|}^2}}}{{\left| {{b_2}} \right|}} \ldots  + \frac{{{h_{1N}}h_{2N}^*{{\left| {{a_N}} \right|}^2}}}{{\left| {{b_N}} \right|}} = 0.
\end{equation}
From \eqref{eq:beam_focus_req}, it can be implied that when the receiver is not at Bob's location, its received signal is almost background noise. That is, the beam focusing effect will appear. Hence, the proof is complete.



\ifCLASSOPTIONcaptionsoff
  \newpage
\fi




\bibliographystyle{IEEEtran}
\bibliography{./bibtex/bib/IEEEabrv,./bibtex/bib/ref}

\begin{IEEEbiography}[{\includegraphics[width=1in,height=1.25in,clip,keepaspectratio]{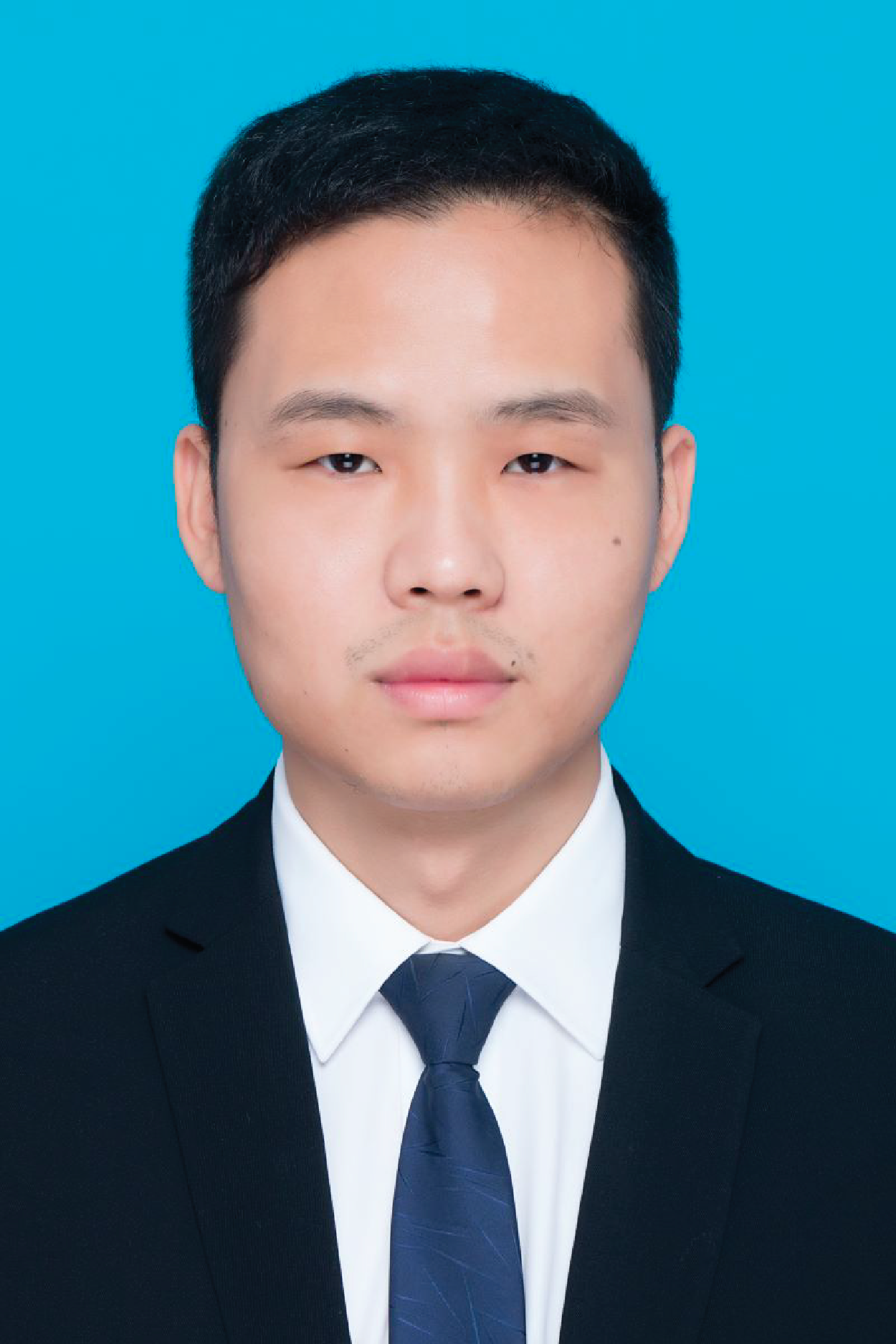}}]{Jun Liu} received the B.S. and M.S. degrees from Central China Normal University, Wuhan, China, in 2015 and 2018, respectively. He is currently pursuing the Ph.D. degree at the University of Electronic Science and Technology of China, Chengdu, China.

His research interests include reconfigurable-intelligent-surface empowered communications and near-field communications.
\end{IEEEbiography}

\begin{IEEEbiography}[{\includegraphics[width=1in,height=1.25in,clip,keepaspectratio]{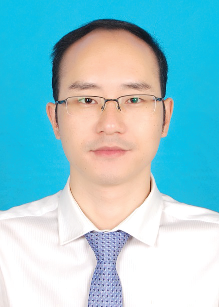}}]{Gang Yang} (S'13-M'15) received the B.Eng. and M.Eng. degrees in communication engineering, communication and information systems, in 2008 and 2011, respectively, from the University of Electronic Science and Technology of China, and the Ph.D. degree from Nanyang Technological University, Singapore, in 2015. He was a Postdoctoral Researcher with the Department of Electrical and Computer Engineering, National University of Singapore, in 2015. He is currently a Full Professor with the Shenzhen Institute for Advanced Study, and the National Key Laboratory of Wireless Communications, University of Electronic Science and Technology of China. His current research interests include next-generation backscatter communications for passive/ambient IoT, reconfigurable-intelligent-surface empowered communications, integrated sensing and communications, and near-field communications.

Dr. Yang was the first-author recipient of the IEEE Communications Society Stephen O. Rice Prize Paper Award in the field of communications theory in 2021, the IEEE Communications Society Transmission, Access, and Optical Systems Technical Committee Best Paper Award in 2016, and the Chinese Government Award for Outstanding Self-Financed Students Abroad in 2015. He is a Top 2\% Scientists by Stanford University. He is the leading editor for the special issue on Integrated Sensing and Communications for 6G IoE of IEEE Internet of Things Journal, an editor for IEEE Communications Letters, an associated editor for IEEE Open Journal of the Communications Society, and an editor for China Communications. He serves for the IEEE GLOBECOMM 2024 as IoT and Sensor Networks Symposium Co-Chair, served for the IEEE GLOBECOMM 2017 as Publicity Co-Chair, and organized three workshops in IEEE VTC-Fall 2022, IEEE VTC-Spring 2021 and IEEE ICCC 2019.
\end{IEEEbiography}

\begin{IEEEbiography}[{\includegraphics[width=1in,height=1.25in,clip,keepaspectratio]{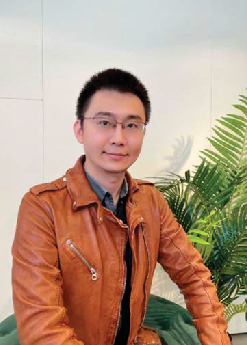}}]{Yuanwei Liu} (S'13-M'16-SM'19-F’24, \url{https://www.eee.hku.hk/~yuanwei/}) has been a (tenured) full Professor in Department of Electrical and Electronic Engineering (EEE) at The University of Hong Kong (HKU) since September, 2024. Prior to that, he was a Senior Lecturer (Associate Professor) (2021-2024) and a Lecturer (Assistant Professor) (2017- 2021) at Queen Mary University of London (QMUL), London, U.K, and a Postdoctoral Research Fellow (2016-2017) at King's College London (KCL), London, U.K. He received the Ph.D. degree from QMUL in 2016.  His research interests include non-orthogonal multiple access, reconfigurable intelligent surface, near field communications, integrated sensing and communications, and machine learning. 

Yuanwei Liu is a Fellow of the IEEE, a Fellow of AAIA, a Web of Science Highly Cited Researcher, an IEEE Communication Society Distinguished Lecturer, an IEEE Vehicular Technology Society Distinguished Lecturer, the rapporteur of ETSI Industry Specification Group on Reconfigurable Intelligent Surfaces on work item of “Multi-functional Reconfigurable Intelligent Surfaces (RIS): Modelling, Optimisation, and Operation”, and the UK representative for the URSI Commission C on “Radio communication Systems and Signal Processing”. He was listed as one of 35 Innovators Under 35 China in 2022 by MIT Technology Review. He received IEEE ComSoc Outstanding Young Researcher Award for EMEA in 2020. He received the 2020 IEEE Signal Processing and Computing for Communications (SPCC) Technical Committee Early Achievement Award, IEEE Communication Theory Technical Committee (CTTC) 2021 Early Achievement Award. He received IEEE ComSoc Outstanding Nominee for Best Young Professionals Award in 2021. He is the co-recipient of the 2024 IEEE Communications Society Heinrich Hertz Award, the Best Student Paper Award in IEEE VTC2022-Fall, the Best Paper Award in ISWCS 2022, the 2022 IEEE SPCC-TC Best Paper Award, the 2023 IEEE ICCT Best Paper Award, and the 2023 IEEE ISAP Best Emerging Technologies Paper Award. He serves as the Co-Editor-in-Chief of IEEE ComSoc TC Newsletter, an Area Editor of IEEE Communications Letters, an Editor of IEEE Communications Surveys \& Tutorials, IEEE Transactions on Wireless Communications, IEEE Transactions on Vehicular Technology, IEEE Transactions on Network Science and Engineering, IEEE Transactions on Cognitive Communications and Networking, and IEEE Transactions on Communications (2018-2023). He serves as the (leading) Guest Editor for Proceedings of the IEEE on Next Generation Multiple Access, IEEE JSAC on Next Generation Multiple Access, IEEE JSTSP on Intelligent Signal Processing and Learning for Next Generation Multiple Access, and IEEE Network on Next Generation Multiple Access for 6G. He serves as the Publicity Co-Chair for IEEE VTC 2019-Fall, the Panel Co-Chair for IEEE WCNC 2024, Symposium Co-Chair for several flagship conferences such as IEEE GLOBECOM, ICC and VTC. He serves the academic Chair for the Next Generation Multiple Access Emerging Technology Initiative, vice chair of SPCC and Technical Committee on Cognitive Networks (TCCN).
\end{IEEEbiography}

\begin{IEEEbiography}[{\includegraphics[width=1in,height=1.25in,clip,keepaspectratio]{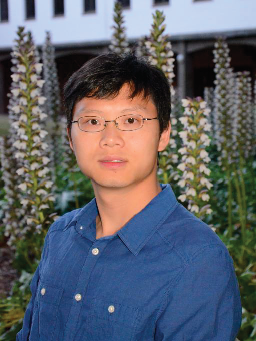}}]{Xiangyun Zhou} (F'23) is an Associate Professor at the Australian National University (ANU). He received the Ph.D. degree from ANU in 2010. His research interests are in the fields of communication theory and wireless networks. He has served as an Editor of IEEE TRANSACTIONS ON COMMUNICATIONS, IEEE TRANSACTIONS ON WIRELESS COMMUNICATIONS and IEEE WIRELESS COMMUNICATIONS LETTERS, and as an Executive Editor of IEEE COMMUNICATIONS LETTERS. He also served as symposium/track and workshop chairs for major IEEE conferences. He is a recipient of the Best Paper Award at ICC'11, GLOBECOM'22, ICC'24 and IEEE ComSoc Asia-Pacific Outstanding Paper Award in 2016. He was named the Best Young Researcher in the Asia-Pacific Region in 2017 by IEEE ComSoc Asia-Pacific Board. He is a Fellow of the IEEE.
\end{IEEEbiography}

\end{document}